\documentclass[lettersize,journal]{IEEEtran}
\IEEEaftertitletext{\vspace{-1.8\baselineskip}}

\usepackage{algorithm}
\usepackage{amsmath,amsfonts,amssymb}
\usepackage{amsthm}
\usepackage{dsfont}
\usepackage{array}
\usepackage[caption=false,font=normalsize,labelfont=sf,textfont=sf]{subfig}
\usepackage{textcomp}
\usepackage{orcidlink}

\usepackage{titlesec}
\titlespacing*{\section}{0pt}{0.6\baselineskip}{0.3\baselineskip}
\titlespacing*{\subsection}{0pt}{0.5\baselineskip}{0.25\baselineskip}
\titlespacing*{\subsubsection}{0pt}{0.4\baselineskip}{0.2\baselineskip}

\usepackage{comment}

\newtheorem{theorem}{Theorem}
\newtheorem{lemma}{Lemma}
\newtheorem{corollary}{Corollary}

\newtheorem{remark}{Remark}
\usepackage{stfloats}
\usepackage{url}
\usepackage{verbatim}
\usepackage{graphicx}
\usepackage{cite}
\usepackage{algpseudocode}
\begin{document}

\title{Characterizing Information Accuracy in Timeliness-Based Gossip Networks}

\author{%
\IEEEauthorblockN{%
Emirhan Tekez\orcidlink{0009-0001-6615-3240}, \IEEEmembership{Graduate Student Member,~IEEE,}
Melih Bastopcu\orcidlink{0000-0001-5122-0642}, \IEEEmembership{Member,~IEEE,}\\
Sinan Gezici\orcidlink{0000-0002-6369-3081}, \IEEEmembership{Fellow,~IEEE}%
}%
\thanks{The authors are with the Department of Electrical and Electronics Engineering, Bilkent University, Ankara, Türkiye (e-mails: emirhan.tekez@bilkent.edu.tr; bastopcu@bilkent.edu.tr; gezici@ee.bilkent.edu.tr).}%

\IEEEcompsocitemizethanks{ 
This work was supported by the TUBITAK 2232-B program (Project No:124C533) and the TUBITAK 2210-A
program.}
}
%

\maketitle

\begin{abstract}
We investigate information accuracy in timeliness-based gossip networks where the source evolves according to a continuous-time Markov chain (CTMC) with $M$ states and disseminates status updates to a network of $n$ nodes. In addition to direct source updates, nodes exchange their locally stored packets via gossip and accept incoming packets solely based on whether the incoming packet is fresher than their local copy. As a result, a node can possess the freshest packet in the network while still not having the current source state. To quantify the amount of accurate information flowing in the network under such a gossiping scheme, we introduce two accuracy metrics, \emph{average accuracy}, defined as the expected fraction of nodes carrying accurate information in any given subset, and \emph{freshness-based accuracy}, defined as the accuracy of the freshest node in any given subset. Using a stochastic hybrid systems (SHS) framework, we first derive steady-state balance equations and obtain matrix-valued recursions that characterize these metrics in fully connected gossip networks under binary CTMCs. We then extend our analysis to the general multi-state information source using a joint CTMC approach. Finally, we quantify the fraction of nodes whose information is accurate due to direct source pushes versus gossip exchanges.  We verify our findings with numerical analyses and provide asymptotic insights.
\end{abstract}

\begin{IEEEkeywords}
information accuracy, gossip networks, misinformation dissemination, timeliness--based gossip
\end{IEEEkeywords}
\vspace{0.1cm}
\section{Introduction}\label{sec:Intro}

Modern networked systems, such as autonomous vehicle fleets \cite{v2x_paper, Kaul2011}, wireless sensor networks \cite{Akkaya2005, Villalba2009, He2020}, and smart factories,  increasingly rely on rapidly-changing information sources. In such settings, 
quick delivery of updates is crucial, which stimulated the rapid growth of the literature on real-time status updates and the \emph{age of information} (AoI) metric \cite{Kaul2011, Yates2019,Yates2020,Sun2022}. AoI quantifies how \emph{old} the latest received update at the receiver is, but it does not directly capture whether the received content is \emph{accurate}. More specifically, for Markov sources, freshness and accuracy can be separate; a newly received update can become incorrect immediately after a source state transition, while an older update may be true if the source's continuous-time Markov chain (CTMC) has reverted back to the relevant state. This has led to accuracy-aware metrics and formulations, including version-based views of Markov source monitoring \cite{Salimnejad2025} and the \emph{age of incorrect information} (AoII) metric \cite{Maatouk2020}, along with policy design under AoII-type objectives \cite{Cosandal2024,Cosandal2025}.

Additionally, gossiping is a well-studied protocol for peer-to-peer exchanges without centralized coordination \cite{Boyd2005,Shah2008,Deb2006,Sanghavi2007}. In the AoI context, \emph{timely gossip} has been studied across different topologies and regimes \cite{Yates2021_SPAWC,  Abd-Elmagid2023,Srivastava2023_grid,Kaswan2023b,Yates_Age_of_Gossip, Kaswan_timestomp,Buyukates2022, Kaswan_mutations, Mitra2023, bastopcu_binary_fresh, cached_harvest}. Here, it refers to a timestamp-based update rule in which a node accepts an incoming packet only if it is fresher than its local copy. Most related works to the problem studied in this manuscript are the recent results on misinformation in gossip networks in \cite{Kaswan_mutations} and \cite{maranzatto_gossip}. The authors of \cite{Kaswan_mutations} study how dishonest nodes and transmission errors affect the propagation of truth under gossiping. In their model, truth serves as a proxy for information accuracy but does not explicitly track node contents, which we do in this manuscript. On the other hand, \cite{maranzatto_gossip} focuses on how misinformation spreads in gossip networks if the transmitted packets' error counts evolve according to a Markov chain as it traverses the network. However, these works do not consider the cases where a source state transition can yield accurate information, even for stale nodes, such as when the source evolves according to a CTMC. 

Our work also connects to a growing line of semantics-aware remote estimation and monitoring for Markov sources, where the \emph{value} of an update depends on what the source state means and what kind of error is being made, rather than on timeliness alone \cite{Luo2024,Luo2025,luo2024exploitingdatasignificanceremote,Saurav2025}. In that literature, the communication decision is typically posed as a possibly state-dependent sampling/query control problem over a direct source–monitor link where one trades off update cost against the risk of incorrect inferences about the Markov state \cite{Luo2024, Luo2025, luo2024exploitingdatasignificanceremote, Saurav2025}. However, such works differ from ours in that gossip networks involve extensive node-to-node packet exchanges and decisions based on local freshness comparisons rather than solely being based on direct source observations. 

Motivated by these considerations, we investigate \emph{information accuracy in timeliness-based gossip networks} driven by an $M$-state CTMC source. We focus on the widely-studied fully-connected (FC) gossip networks \cite{Yates_Age_of_Gossip,Buyukates2022,Kaswan_mutations,Kaswan_survey, Kaswan_timestomp, bastopcu_binary_fresh}. The source pushes its current state to nodes, while nodes gossip their stored packets and keep an incoming packet only if it is \emph{fresher} than the local copy, as in version-based dissemination \cite{Abolhassani2021,Buyukates2022,Salimnejad2025}. Since packet acceptance at a node is based solely on timeliness, a node may hold the freshest packet available in the network while still being incorrect with respect to the current source state. We quantify this effect via two metrics defined for arbitrary $k$-node subsets: $(i)$ \emph{average accuracy}, the expected fraction of nodes holding accurate information in a subset, and $(ii)$ \emph{freshness-based accuracy}, the accuracy of the \emph{freshest} node in a subset.

Our contributions can be listed as follows: $(i)$ We perform a quantitative analysis of how accurate and inaccurate information spreads in timeliness-based gossip networks using the stochastic hybrid systems (SHS) approach \cite{Hespanha2007, yates_aoi_moments}. $(ii)$ We develop matrix-valued recursive equations that quantify the fraction of nodes holding accurate information within the network under binary Markov sources. $(iii)$ We extend our results to a finite $M$-state CTMC using a joint CTMC approach. $(iv)$ We quantify the fraction of nodes that are accurate due to \emph{source pushes} and the fraction of nodes that are accurate due to \emph{source inversions}. Notably, this is the first work to fully characterize information accuracy in timeliness-based gossip networks while explicitly accounting for the nodes' information contents, without using proxies to simplify analysis. 
$(v)$ We provide insights into how information accuracy in gossip networks behaves under asymptotic regimes. We show that in timeliness-based gossip networks, increasing the peer-to-peer gossip rate is far less effective at improving accuracy than increasing the source push rate. $(vi)$ Additionally, we prove that the number of nodes holding any given information in the network is only dependent on the stationary distribution of the source CTMC, and not on any transmission rates. In addition, we support our findings via simulation results for various source push and gossip rates.

\vspace{-0.1cm}

\section{System Model and Problem Formulation}\label{sect:model}

We consider a fully-connected (FC) gossip network with $n$ nodes,
$\mathcal{N} = \{1,\ldots,n\}$, and a source that disseminates its most recent status updates. The information at the source evolves according to an irreducible, finite-state continuous-time
Markov chain (CTMC) having a total of $M$ states, that is,
\( Q(t) \in \{1,\ldots,M\}\), and updates itself according to a known generator matrix, $\mathbf{Q}$, with entries labeled $q_{ij}$ to denote transitions from state $i$ to state $j$. The source sends updates (the current state of the CTMC) to each node in the network, $i \in \mathcal{N}$, according to a thinned Poisson process with rate $\frac{\lambda_s}{n}$. We assume the source to be always holding fresh and true information.  A simple example of the system of interest is illustrated in Fig.~\ref{fig:system_model} for $n=6$ nodes.

\begin{figure}[t]
    \centering
    \includegraphics[width=0.83\columnwidth]{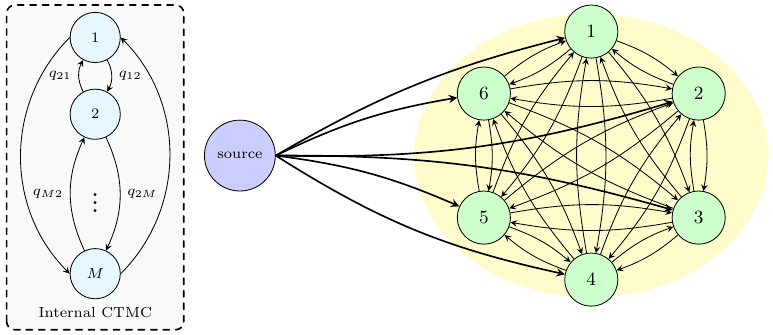}
    \vspace{-0.4cm}
    \caption{A source with $M$-state CTMC connected to a fully-connected gossip network of $n=6$ nodes.}
    \label{fig:system_model}
    \vspace{-0.65cm}
\end{figure}

When the state of the CTMC changes, a new version of the update is generated at the source. We denote the most recent version of the information at the source at time $t$ as $V_0(t)$. Thus, when a new version of the update is generated, the version of the information at the source, $V_0(t)$, increases by one. Similarly, we denote the most recent version of the update at node $i$ at time $t$ as $V_i(t)$ and its content as $S_i(t)$. When the source sends its information to node $i$ at time $t$, node $i$ receives both the current state of the CTMC and also its version information in the form of $\{V_0(t), Q(t)\}$. The nodes would like to keep track of the current state of the CTMC as \textit{accurately} as possible. For that, the nodes always accept updates from the source, as the source maintains the most recent and accurate information. Thus, if node $i$ receives an update from the source at time $t$, the node updates its information as $V_i(t) = V_0(t)$ and $S_i(t) = Q(t)$. In order to measure information freshness at the nodes, we use a metric called version age of information (VAoI) introduced in \cite{Yates_Age_of_Gossip,Abolhassani2021}. The version age of node $i$ at time $t$ is denoted by $X_i(t)$ and is given by
\begin{align}\nonumber\\[-21pt]
    X_i(t) = V_0(t)-V_i(t).\\[-21pt]\nonumber
\end{align}
Since the source has a finite update rate $\lambda_s$ that is equally shared among the nodes, increasing the network size may render the source updates insufficient to maintain accurate information at all nodes. In order to improve their information freshness, the nodes also share the most recent information through the gossiping protocol \cite{Yates2021_SPAWC, Yates_Age_of_Gossip}. For that, each node in the network sends updates to one another according to another thinned Poisson process, where node-to-node gossip takes place with rate $\frac{\lambda}{n-1}$. When two nodes gossip with one another, they only consider the version age of the incoming packet as they do not know whether the incoming packet is accurate or not. More formally, when node $i$ tries to send a packet to node $j$, the version age at node $j$ evolves according to the following rule:
\begin{align}\nonumber\\[-21pt]
X_j(t) = \min\{X_i(t^-), X_j(t^-) \},
\label{eq:x_j_piecewise}\\[-21pt]\nonumber
\end{align}
where $t^-$ represents the time instant just prior to the information exchange. Thus, when node $i$ sends its information to node $j$, node $j$ checks the incoming information version and if the incoming information is more recent, that is, $V_i(t)>V_j(t)$, node $j$ accepts node $i$'s information; otherwise, it continues to keep its current information. As a result, the version age at node $j$ is given in (\ref{eq:x_j_piecewise}). In order to analyze the average version age, we need to keep track of the version age of an arbitrary subset of $k\!$ elements denoted by $\!X_{A_k}\!\!$ which is given by
\begin{equation}
    X_{A_k}(t) \triangleq \min_{j \in A_k} X_j(t).
    \label{eq:version_age_subset}
\end{equation}
In other words, the version age of a set is determined by the freshest information in set $A_k$. 

In this work, our goal is to characterize the average information accuracy of the nodes in the network. With this goal, we define the accuracy of the node $i$ denoted by $C_i(t)$ as:
\begin{align} \nonumber\\[-21pt]
C_i(t) =\begin{cases} 
      1, & S_i(t) = Q(t), \\
      0, & S_i(t) \neq Q(t). 
   \end{cases}
\label{eq:c_j_indicator}\\[-21pt]\nonumber
\end{align}
With the definition in (\ref{eq:c_j_indicator}), accuracy is a binary metric that takes the value $1$ when the information at node $i$ is the same as the information at the source, and $0$ otherwise. This notion of accuracy was investigated under the name \emph{binary freshness} in \cite{nail_hoca_freshness, nail_hoca_20_ocak}, though these works deal with a single source monitor, and do not consider how the metric evolves for gossip networks that include vast amounts of peer-to-peer exchanges. Similarly, we define $C_{A_k}(t)$ as the average accuracy of an arbitrary subset of the network with $k$ nodes, denoted by $A_k$, where $C_{A_k}(t)$ is given by
\begin{align}\nonumber\\[-21pt]
C_{A_k}(t) = \frac{C_1(t) + C_2(t)+\cdots+C_k(t)}{k} = \frac{1}{k} \sum_{\ell \in A_k} C_\ell.
\label{eq:c_ak_defn}\\[-21pt]\nonumber
\end{align}

Note that the accuracy definition depends only on the contents of the source and the relevant node, and does not depend on the version age of the node. Thus, in general, the nodes can have the accurate information in two ways: $i)$ a node obtains the most recent information at the source and as a result, its version age is equal to 0 and has the accurate information, and $ii)$ a node may have an older version of the information, but this stale packet can be accurate if the source returns back to the relevant state. In other words, a stale node carrying an older version of the information may have the correct information which differentiates information timeliness from the information accuracy. We also note that the information accuracy metric that we defined in~\eqref{eq:c_j_indicator} is a special case of the more general metric called AoII defined in \cite{Maatouk2020}. Finally, during gossiping between nodes, as we mentioned earlier, the nodes do not know whether they hold the accurate information or not. Because of this reason, information exchanges still happen based on the information freshness as described in (\ref{eq:x_j_piecewise}). Thus, when node $i$ sends its status update to node $j$, the information accuracy at node $j$ changes as follows:   
\begin{align}
C_j(t) =
\begin{cases}
    C_i(t^-), & X_i(t^-) \le X_j(t^-), \\
    C_j(t^-), & X_i(t^-) >  X_j(t^-).
\end{cases}
\label{eq:c_j_piecewise}\\[-21pt]\nonumber
\end{align}
In other words, if node $i$ carries fresher information, then node $j$'s accuracy is updated to match that of node $i$; otherwise, it remains unchanged.

In the following sections, we analyze the information accuracy of the nodes in the network for a binary information source (i.e., $M=2$), as well as for a source with a general number of states (i.e., $M>2$).

\section{The Binary Information Source}\label{sect:binary}

We first study the average accuracy under a binary CTMC source model (i.e., $M=2$) with a generator matrix given by
\begin{align}\nonumber\\[-21pt]\label{eq:q_binary_structure}
    \mathbf{Q}=\begin{bmatrix} q_{11} & q_{12} \\[2pt] q_{21} & q_{22} \end{bmatrix},\\[-21pt]\nonumber
\end{align}
where $q_{11}=-q_{12}$ and $q_{22}=-q_{21}$. In other words, the source transitions from state $1$ to $2$ at rate $q_{12}$, and from state
$2$ to $1$ at rate $q_{21}$. The stationary distribution of the binary CTMC is
\begin{align} \nonumber\\[-21pt]
    \pi_1=\frac{q_{21}}{q_{12}+q_{21}}, 
    \qquad
    \pi_2=\frac{q_{12}}{q_{12}+q_{21}}\cdot
\label{eq:binary_stationary_pi}\\[-21pt]\nonumber
\end{align}

In order to evaluate the expected number of nodes that hold correct information, \(c_k = \lim_{t \to \infty} \mathbb{E}[C_{A_k}(t)],\)
we model the system using the stochastic hybrid system (SHS) framework \cite{Hespanha2007, yates_aoi_moments}. The hybrid state is given by
\begin{align}\nonumber\\[-21pt]
  \bigl(Q(t),\mathbf{Z}(t)\bigr)
=
\bigl(Q(t),\mathbf{C}(t),\mathbf{X}(t)\bigr),  
\label{eq:hybrid_state}\\[-21pt]\nonumber
\end{align}
where the discrete mode $Q(t)\in\{1,2\}$ denotes the source's CTMC state, and the
continuous component consists of 
$\mathbf{C}(t) = [C_1(t), \ldots, C_n(t)] \in \{0,1\}^n \subset \mathbb{R}^n$ where $C_i(t)$ represents the accuracy of node $i$ at time $t$, and $\mathbf{X}(t) = [X_1(t), \ldots, X_n(t)] \in \mathbb{R}^n,$ (denoting the age vector of all the nodes in the network) collected as $\mathbf{Z}(t) = (\mathbf{C}(t),\mathbf{X}(t))\in \mathbb{R}^{2n}$.
Between jumps, the continuous state is constant, i.e., it obeys the differential equation
$(\dot{\mathbf{C}}(t),\,\dot{\mathbf{X}}(t)) = \mathbf{0}_{2n}$.

The set of possible transitions for this SHS is given by
\begin{align}\nonumber\\[-21pt]\label{eq:set_of_transitions}
\mathcal{E}
&= \{(0,0,1\!\to\!2),(0,0,2\!\to\!1)\}
\;\cup\; \{(0,j):\, j\in\mathcal{N}\}\\
&\quad\cup\; \{(i,j):\, i,j\in\mathcal{N},\, i\neq j\}.\nonumber\\[-21pt]\nonumber
\end{align}
where node $0$ represents the source. Transitions $(0,0,1\to 2)$ and $(0,0,2\to 1)$ represent source flips from state
$1$ to $2$ and from state $2$ to $1$, respectively. The transition $(0,j)$
represents the case where the source sends an update to node $j$, and transition $(i,j)$ represents a gossiping event from node $i$ to node $j$. Note that in this work, we assume that gossip events are directed; for each ordered pair $(i,j)$ with $i\neq j$, node $i$ sends a packet to node $j$, and only node $j$'s accuracy and version-age can change under the described gossip event. Each transition $e\in\mathcal{E}$ resets the hybrid state at time $t$ according to
\begin{align}\nonumber\\[-20pt]
\bigl(Q(t^+),\mathbf{Z}(t^+)\bigr)
=
\phi_e\bigl(Q(t^-),\mathbf{Z}(t^-)\bigr),
\label{eq:q_z_reset}\\[-20pt]\nonumber
\end{align}
where, for brevity, we only present the $\mathbf{Z}$-component in \eqref{eq:X_update_nodewise}--\eqref{eq:F_Ak_update} below. For the transitions specified in~\eqref{eq:set_of_transitions}, the mode-dependent rate map $\lambda_e(q)$ is given by
\begin{align}\nonumber\\[-21pt]
\lambda_{e}(q) =
\begin{cases}
q_{12}\,\mathds{1}_{\{q=1\}}, & e = (0,0,1\to 2), \\
q_{21}\,\mathds{1}_{\{q=2\}}, & e = (0,0,2\to 1), \\
\dfrac{\lambda_s}{n},      & e = (0,j),\; j\in\mathcal{N}, \\
\dfrac{\lambda}{n-1},      & e = (i,j),\; i,j\in\mathcal{N},\; i\neq j,
\end{cases}
\label{eq:lambda_map}\\[-21pt]\nonumber
\end{align}
where $\mathds{1}_{\{x\}}$ is the indicator function taking value $1$ only when $x$ is true and value 0, otherwise. We note from (\ref{eq:lambda_map}) that the  transition $(0,0,1\to 2)$ is active only in mode $Q=1$ with rate $q_{12}$,
and transition $(0,0,2\to 1)$ is active only in mode $Q=2$ with rate $q_{21}$, while push
and gossip rates do not depend on the source mode.

For each node $\ell \in \mathcal{N}$, the state $(\mathbf{C}(t),\mathbf{X}(t))$ evolves according to the following reset maps under the transitions specified in~\eqref{eq:set_of_transitions}. First, we provide the version age reset map for a single node $\ell \in \mathcal{N}$ as follows:
\begin{align} \nonumber\\[-21pt]
\!\!\!\!X_\ell' \!\!=\!\!
\begin{cases}
X_\ell + 1,
& e \!\in\! \{(0,0,1\!\to\! 2), (0,0,2\!\to\! 1)\}, \ell \!\in \!\mathcal{N},\!\!\!\!\!\!\!\!\!\!\!\!\!
\\
0,
& e = (0,j),\; \ell = j \in \mathcal{N},
\\
\min\bigl(X_i, X_j\bigr),
& e = (i,j),\; i,j \in \mathcal{N},\; i\neq j,\; \ell = j,\!\!\!\!\!
\\
X_\ell,
& \text{otherwise.}
\end{cases}
\label{eq:X_update_nodewise}\\[-21pt]\nonumber
\end{align}
For any source inversion event, the version age of information increases by $1$ for each node since the source generates a fresh packet at each state transition. Under gossiping events, the receiver takes on the sender's version age only if the incoming packet is fresher as shown in~\eqref{eq:x_j_piecewise}.

Similarly, the accuracy indicator for node $\ell \!\in\! \mathcal{N}$ evolves as:
\begin{align} \nonumber\\[-21pt]
\!\!\!\!C'_\ell \!\!=\!\!
\begin{cases}
1 - C_\ell,
& \!\!\!e \!\in\! \{(0,0,1\!\to\! 2),\! (0,0,2\!\to\! 1)\}, \ell \!\in\! \mathcal{N},\!\!\!\!\!\!\!\! \!\!\!\!\!\!\!\!\\
1,
& \!\!\!e = (0,j),\; \ell = j \in \mathcal{N}, \\
C_{\operatorname*{arg\,min}_{\ell \!\in\! \{i,j\}} \!X_\ell},
& \!\!\!e = (i,j),\; i,j \in \mathcal{N},\; i \neq j,\; \ell\! \!=\!\! j,\!\!\!\!\!\!\!\!\!\!\! \\
C_\ell,
& \!\!\!\text{otherwise}.
\end{cases}
\label{eq:C_update_nodewise}\\[-21pt]\nonumber
\end{align}
For every source inversion, the accuracy of any node is inverted as provided in the first case of  \eqref{eq:C_update_nodewise}. Since we assume that the source always holds true and fresh information, $X_\ell$ and $C_\ell$ are set to $0$ and $1$, respectively, after receiving information from the source. The third case in~\eqref{eq:C_update_nodewise} corresponds to a case where node $i$ sends an update to node $j$. As we mentioned earlier, since the nodes are unaware of the accuracy of their information, node $j$ decides whether to accept or reject node $i$'s update based on the information's freshness. As a result of the third case in~\eqref{eq:C_update_nodewise}, the accuracy of node $j$ becomes the accuracy of the freshest information in the set containing nodes $i$ and $j$.     

We can extend the node-wise reset maps in~\eqref{eq:X_update_nodewise} and~\eqref{eq:C_update_nodewise} to arbitrary subsets of the network containing $k$ nodes to investigate the $C_{A_k}$ and $X_{A_k}$ quantities described in Section~\ref{sect:model}. First, we provide version age reset maps for an arbitrary set $A_k$ as follows:
\begin{align}\nonumber\\[-21pt]
\!\!X'_{A_k} \!\!=\!\!
\begin{cases}
X_{A_k} + 1,
& e \in \{(0,0,1\to 2),\, (0,0,2\to 1)\}, \\
0,
& e = (0,j),\; j \in A_k, \\
X_{A_{k+1}},
& e = (i,j),\; i \in \mathcal{N}\setminus A_k,\; j \in A_k, \\
X_{A_k},
& \text{otherwise.}
\end{cases}
\label{eq:X_Ak_update}\\[-21pt]\nonumber
\end{align}
We note that the version age of set $A_k$, i.e., $X_{A_k}(t)$, is defined in (\ref{eq:version_age_subset}). When the information at the source is updated, the version age of the freshest node within set $A_k$, thus the version age of set $A_k$, increases by one which is the first case provided in (\ref{eq:X_Ak_update}). Whenever the source sends an update to a node within set $A_k$, the recipient node $j$ becomes the freshest node within set $A_k$, resetting the version age of the set to 0. If node $i$ outside of set $A_k$ sends an update to a node within set $A_k$, the receiver node will copy the version-age of the outsider node if the incoming packet is fresher. Thus, we need to investigate the minimum age of the super-sized set $A_k \; \cup \; \{i\}$ to determine the post-transition value. Since the network is fully connected, we can say $X_{A_k \; \cup \; \{i\}} = X_{A_{k+1}}$ by symmetry. Similarly, we provide the reset map for the accuracy of an arbitrary set $A_k$ as follows:   
\begin{align}\nonumber\\[-21pt]
\!\!\!\!C'_{A_k}\!\! =\!\!
\begin{cases}
1 - C_{A_k},
&\!\! e \in \{(0,0,1\!\to\! 2),\! (0,0,2\!\to\! 1)\},\!\!\!\!\!\!\!\!\!\! \\
C_{A_k} + \dfrac{1 - C_j}{k},
& \!\!e = (0,j),\; j \in A_k, \\
C_{A_k} \!+\! \dfrac{F_{\{i,j\}} \!- \!C_j}{k},
& \!\!e \!=\! (i,j), i,j \!\in \!\mathcal{N},\; i \!\neq \!j,\; j \!\in\! A_k, \!\!\!\!\!\!\!\!\!\!\!\!\!\!\!\!\!\!\!\!\\
C_{A_k}.
&\!\! \text{otherwise.}
\end{cases}
\label{eq:C_Ak_update}\\[-21pt]\nonumber
\end{align}
Whenever the information at the source changes, the accuracy of set $A_k$ flips as provided in the first case of (\ref{eq:C_Ak_update}). When the source sends an update to node $j$ in set $A_k$, the accuracy increases by $\frac{1}{k}$ if node $j$'s accuracy is 0. Similarly, if node $i$ that is not in set $A_k$ sends an update to node $j$ that is in set $A_k$, node $j$ accepts the update if node $i$ carries a fresher update. In this case, the accuracy of the set $A_k$ $i)$ increases by $\dfrac{1}{k}$ when node $i$ carries accurate information and node $j$ has inaccurate information; $ii)$ stays the same if both nodes' accuracy is the same;  $iii)$ decreases by $\dfrac{1}{k}$ if node $i$ has inaccurate information and node $j$ has accurate information. Such an operation in the third case requires us to define a new function denoted by $F_{A_k}$ which is the accuracy of the freshest node in an arbitrary set $A_k$ given by 
\begin{align}\nonumber\\[-21pt]
  F_{A_k}(\mathbf{C},\mathbf{X}) \triangleq
C_{\arg\min_{\ell\in A_k} X_\ell}.  
\label{eq:f_ak_defn}\\[-21pt]\nonumber
\end{align}
This corresponds to a test function of the form
\begin{align}\nonumber\\[-21pt]
  \psi_{k}(Q,\mathbf{Z}) = F_{A_k}(\mathbf{C},\mathbf{X}),
\label{eq:formal_fk_defn}\\[-21pt]\nonumber
\end{align}
which does not depend explicitly on the discrete mode $Q$ and is constant
between jumps. We next define the reset map which governs how $F_{A_k}$ behaves under the transitions in~\eqref{eq:set_of_transitions}, i.e., $F'_{A_k} = \psi_{k}\bigl(\phi_e(Q,\mathbf{Z})\bigr)$:
\begin{align}\nonumber\\[-21pt]
\!\!F'_{A_k}\! \!=\!\!
\begin{cases}
1 - F_{A_k},
& e \in \{(0,0,1\to 2),\, (0,0,2\to 1)\}, \\
1,
& e = (0,j),\; j \in A_k, \\
F_{A_{k+1}},
& e = (i,j),\; i \in \mathcal{N}\setminus A_k,\; j \in A_k, \\
F_{A_k},
& \text{otherwise.}
\end{cases}
\label{eq:F_Ak_update}\\[-21pt]\nonumber
\end{align}
When information at the source is updated, $F_{A_k}$ flips, similar to the first case defined in~\eqref{eq:C_Ak_update}. When the source sends an update to a node in set $A_k$, the receiver node becomes the freshest node within the set and inherits the source's accuracy, hence $F_{A_k}'=1$. During gossip events where a node outside $A_k$ gossips to one inside, the new freshness-based accuracy of $A_k$ coincides with that of the expanded set $A_k\cup\{i\}$ and, by symmetry, equals $F_{A_{k+1}}$.

We recall the steady-state SHS balance identity. For any (time-invariant) test
function $\psi$ of the hybrid state, we have
\begin{subequations}\label{eq:shs_generator_balance}
\begin{align}\nonumber\\[-21pt]\label{eq:shs_stationarity}
\frac{d}{dt}\,\mathbb{E}\!\left[\psi\bigl(Q(t),\mathbf{Z}(t)\bigr)\right] = 0 .\\[-21pt]\nonumber
\end{align}
since the test function’s expected value does not change with time under steady state conditions. The extended generator $\mathcal{L}\psi(Q,\mathbf{Z})$ gives the instantaneous change of
$\psi(Q(t),\mathbf{Z}(t))$ by summing the jump contributions over all transition types $e\!\in\!\mathcal{E}$:
\begin{align}\nonumber\\[-21pt]\label{eq:shs_generator_def}
\mathcal{L}\psi(Q,\mathbf{Z})
=
\sum_{e\in\mathcal{E}}
\lambda_e(Q)
\Bigl[
\psi\bigl(\phi_e(Q,\mathbf{Z})\bigr)
-
\psi(Q,\mathbf{Z})
\Bigr] .\\[-21pt]\nonumber
\end{align}
Using the fact that the expected change in steady-state is zero, we obtain
\begin{align}\label{eq:shs_balance_equation}\nonumber\\[-21pt]
\!\!\!\sum_{e\in\mathcal{E}}
\!\mathbb{E}\!\left[\!
\lambda_e\bigl(Q(t)\bigr)
\!\Bigl(\!
\!\psi\bigl(\!\phi_e(Q(t),\!\mathbf{Z}(t))\bigr)
\!-\!
\psi\bigl(Q(t),\!\mathbf{Z}(t)\!\bigr)
\!\Bigr)
\!\right]\!=\!0 .\!\!
\end{align} \\[-21pt]\nonumber
\end{subequations}

For the particular choices of $\psi(\cdot)$ introduced below (e.g., $F_{A_k}$ and $C_{A_k}$'s mode-tagged variants), \eqref{eq:shs_balance_equation} yields the SHS balance
equations.

\subsection{Mode-Tagged Freshness Accuracy Characterization}
Since the rates of source push events depend on which discrete mode the system was in before transition $e$, we first derive an exact characterization of the freshness-based accuracy $F_{A_k}$ that explicitly tracks the source mode $Q(t)$. For $m\!\in\!\{1,2\}$ and $k\!\!\in\!\!\{1,\ldots,n\}$, we define the mode-tagged expectations as
\begin{align}\nonumber\\[-21pt]
f_k^{(m)}
\;\triangleq\;
\lim_{t \to \infty}\mathbb{E}\!\left[F_{A_k}(t)\,\mathds{1}_{\{Q(t)=m\}}\right],
\label{eq:fkm_def}\\[-21pt]\nonumber
\end{align}
and also $\mathbf{f}_k \triangleq  [f_k^{(1)}, f_k^{(2)}]^T$ where $[\cdot]^T$ denotes the transpose of a vector. The sum of these two quantities yield the unconditional freshness-based accuracy given by:
\begin{align}\nonumber\\[-21pt]
f_k \triangleq f_k^{(1)} + f_k^{(2)}
= \lim_{t\to\infty} \mathbb{E}[F_{A_k}(t)].
\label{eq:fk_uncond}\\[-21pt]\nonumber
\end{align}
To obtain the balance equations, we use the test functions
\begin{align}\nonumber\\[-21pt]
\psi_k^{(1)}(Q,\mathbf{Z})
&= F_{A_k}(\mathbf{C},\mathbf{X})\,\mathds{1}_{\{Q=1\}},\label{eq:psi_fk_modes}
\\[-3pt]
\psi_k^{(2)}(Q,\mathbf{Z})
&= F_{A_k}(\mathbf{C},\mathbf{X})\,\mathds{1}_{\{Q=2\}}.\\[-21pt]\nonumber
\end{align}
where $F_{A_k}(\mathbf{C},\mathbf{X})$ is the accuracy of the freshest node in set $A_k$ defined in (\ref{eq:f_ak_defn}) and we have $f_k^{(m)} = \lim_{t \to \infty} \mathbb{E}[\psi_k^{(m)}]$. For notational convenience, we define the rates
\begin{align}
\alpha_k \triangleq \frac{k\lambda_s}{n},
\quad
\beta_k \triangleq \frac{k(n-k)\lambda}{n-1},
\quad
\gamma_k \triangleq \frac{k(k-1)\lambda}{n-1},
\label{eq:alpha_beta_def} \\[-21pt]\nonumber
\end{align}
which are, respectively, the total source push rate to the nodes in set $A_k$, the
rate of gossip events from nodes outside $A_k$ into nodes inside $A_k$, and the rate of gossip events from and to nodes inside $A_k$. We note that $\alpha_k$, $\beta_k$, and $\gamma_k$ do not depend on the discrete mode $Q$.

Let us consider the test function $\psi_k^{(1)}(Q,\mathbf{Z})= F_{A_k}\mathds{1}_{\{Q(t)=1\}}$ given in (\ref{eq:psi_fk_modes}). There are three types of events that can induce a change: $a)$ source inversions, $b)$ source pushes, and $c)$ gossiping events. As a result, the SHS balance equation for $\psi_k^{(1)}(Q,\mathbf{Z})$ can be written as: 
\begin{align}\label{eq:SHS_F_A_k}
\lim_{t\to\infty}\!\mathbb{E}\Big[\!
&\Big(\!\!-\!\!q_{12}F_{A_k}\!\!
+\!\alpha_k(1\!-\!F_{A_k})
\!\!+\!\!\beta_k(F_{A_{k+1}}\!\!-\!F_{A_k})\!\Big)\notag\\[-0.2em]
&\cdot \mathds{1}_{\{Q=1\}}
+\Big(\!q_{21}(1\!-\!F_{A_k})\!\Big)\mathds{1}_{\{Q=2\}}
\!\Big]=0.
\end{align}
The first term in (\ref{eq:SHS_F_A_k}) comes from the source inversion transition $(0,0,1\to2)$ which happens with rate $q_{12}\mathds{1}_{\{Q(t)=1\}}$.
After this transition, the source state becomes $Q'(t)=2$. From~\eqref{eq:F_Ak_update}, the
freshness-based accuracy for state~1 becomes $0$ (since the source is now in state~2), hence
the overall change induced by this transition can be written as $(0 - F_{A_k})\mathds{1}_{\{Q=1\}}$. Similarly, the second term in (\ref{eq:SHS_F_A_k}) is obtained for $(0,0,2\to1)$ transition, which occurs with rate
$q_{21}\mathds{1}_{\{Q(t)=2\}}$. In this case, after the transition the new source state becomes
$Q'(t)=1$, and~\eqref{eq:F_Ak_update} implies the freshness-based accuracy for state~1 becomes $1-F_{A_k}$. Thus, we can express the overall change for this transition as $(1 - F_{A_k})\mathds{1}_{\{Q=2\}}$. The third term in (\ref{eq:SHS_F_A_k}) represents source updating a node in $A_k$ in which case  $F_{A_k}$ takes the value $1$. The overall change induced by this transition can be written as $(1 - F_{A_k})\mathds{1}_{\{Q=1\}}$. Finally, the last term in (\ref{eq:SHS_F_A_k}) represents the case where a node outside $A_k$ gossips to a node inside $A_k$ which happens with update rate $\beta_k$. After this update, the accuracy of set $A_k$ becomes the freshness-based accuracy of the super-set $F_{A_{k+1}}$ that includes the sender node. Thus, we can write the change induced by this transition as $(F_{A_{k+1}} - F_{A_k})\mathds{1}_{\{Q=1\}}$. 

After taking the expectations in (\ref{eq:SHS_F_A_k}), we obtain 
\begin{equation}
- q_{12} f_k^{(1)} + q_{21}(\pi_2 - f_k^{(2)})
+ \alpha_k(\pi_1 - f_k^{(1)})
+ \beta_k(f_{k+1}^{(1)} - f_k^{(1)})= 0,
\label{eq:F_balance_mode1}
\end{equation}
for $k=1,\ldots,n-1$. After repeating a similar analysis for $\psi_k^{(2)}(Q,\mathbf{Z})= F_{A_k}\mathds{1}_{\{Q(t)=2\}}$, we can obtain: 
\begin{equation}
q_{12}(\pi_1 - f_k^{(1)})
- q_{21} f_k^{(2)}
+ \alpha_k(\pi_2 - f_k^{(2)})
+ \beta_k(f_{k+1}^{(2)} - f_k^{(2)}) = 0,
\label{eq:F_balance_mode2}
\end{equation}
for $k=1,\ldots,n-1$. Combining (\ref{eq:F_balance_mode1}) and (\ref{eq:F_balance_mode2}), we obtain a linear recursion of the form
\begin{equation}
W_k \mathbf{f}_k = \mathbf{v}_k + V_k \mathbf{f}_{k+1},
\qquad
k=1,\ldots,n-1,
\label{eq:fk_matrix_recursion}
\end{equation}
with
\begin{align}\label{eq:A_k_B_k_def}
\resizebox{\columnwidth}{!}{$
W_k \!\triangleq\!
\begin{bmatrix}
q_{12}\!+\!\alpha_k\!+\!\beta_k & q_{21}\\
q_{12} & q_{21}\!+\!\alpha_k\!+\!\beta_k
\end{bmatrix},
~
V_k\! \triangleq\! \beta_k I_2,
~
\mathbf{v}_k \triangleq
\begin{bmatrix}
q_{21}\pi_2\!+\!\alpha_k\pi_1\\
q_{12}\pi_1\!+\!\alpha_k\pi_2
\end{bmatrix},
$}
\end{align}
where $I_2$ represents a $2\times2$ identity matrix. Since $W_k$ is an invertible matrix for $\alpha_k, \beta_k >0$, we obtain $\mathbf{f}_k$ as
\begin{equation}
\mathbf{f}_k
=
W_k^{-1}\bigl(\mathbf{v}_k + V_k\mathbf{f}_{k+1}\bigr),
\qquad
k=1,\ldots,n-1.
\label{eq:fk_backward_recursion}
\end{equation}
We note that (\ref{eq:fk_backward_recursion}) implies a matrix-valued backward recursion for the mode-tagged freshness-based accuracy.

When $k=n$, the set $A_n$ contains all nodes. Hence,
$\beta_n = 0$ and $\alpha_n = \lambda_s$. The balance equations
\eqref{eq:F_balance_mode1}–\eqref{eq:F_balance_mode2} reduce to
\begin{align}
q_{21}(\pi_2 - f_n^{(2)})
- q_{12} f_n^{(1)}
+ \lambda_s(\pi_1 - f_n^{(1)}) =& 0,
\label{eq:F_balance_kn_1}
\\[4pt]
q_{12}(\pi_1 - f_n^{(1)})
- q_{21} f_n^{(2)}
+ \lambda_s(\pi_2 - f_n^{(2)}) =&0.
\label{eq:F_balance_kn_2}
\end{align}
Solving the $2\times 2$ system in~\eqref{eq:F_balance_kn_1}–\eqref{eq:F_balance_kn_2} along with $(\pi_1, \pi_2)$ in (\ref{eq:binary_stationary_pi}) yields
\begin{align}\label{eq:boundary_solution}
\!\!\mathbf{f}_n
\!\triangleq\!
\begin{bmatrix}
f_n^{(1)}\\[2pt]
f_n^{(2)}
\end{bmatrix}
\!=\!
\frac{1}{(q_{12} \!+\! q_{21})(\lambda_s\!+\!q_{12} \!+\! q_{21})}
\begin{bmatrix}
q_{21}(\lambda_s\!+\!q_{21})\\
q_{12}(\lambda_s\!+\!q_{12})
\end{bmatrix}\!.\!\!
\end{align}
Given $\mathbf{f}_n$, the recursion in \eqref{eq:fk_backward_recursion} determines $\mathbf{f}_k$ for all $k=n-1,\cdots,1$. Having completed the characterization of freshness-based accuracy, which is required to analyze average accuracy, we now proceed to study the average accuracy of the network. 

\subsection{Characterization of Mode-Tagged Average Accuracy}\label{sec:mode_tagged_acc}

In (\ref{eq:C_Ak_update}), we show how the accuracy of set $A_k$ denoted by $C_{A_k}$ resets based on the possible transitions in the network. We see that $C_{A_k}$ depends on the accuracy of the freshest node in set $A_k$ denoted by $F_{A_k}$. We will now proceed to provide an explicit characterization of the mode-tagged average accuracy to provide a complete analysis, and derive expressions which will be useful for providing insights into how accurate and inaccurate information spreads under asymptotic regimes. 

\begin{theorem}\label{thm:avg_c_theorem}
Let $\mu\triangleq \lambda_s/n$ denote the per-node source push rate and let
$\mathbf{c}\triangleq [c^{(1)},\,c^{(2)}]^T$ denote the vector consisting of the mode-tagged average accuracies.
Then, $\mathbf{c}$ is the unique solution of the linear system
\begin{align}\label{eq:c_linear_system}\nonumber\\[-21pt]
\bar W\,\mathbf{c}=\mathbf{b}+\lambda\,\mathbf{f}_2,\\[-21pt]\nonumber
\end{align}
where
\[
\bar W \triangleq
\begin{bmatrix}
q_{12}+\mu+\lambda & q_{21}\\
q_{12} & q_{21}+\mu+\lambda
\end{bmatrix},~
\mathbf{b}\triangleq
\begin{bmatrix}
q_{21}\pi_2+\mu\pi_1\\
q_{12}\pi_1+\mu\pi_2
\end{bmatrix},
\]
and $ \mathbf{f}_2\triangleq
\begin{bmatrix}
f_2^{(1)},
f_2^{(2)}
\end{bmatrix}^T.$ Consequently, the unconditional average accuracy for any arbitrary subset of the network with cardinality $k \in \{1,2,\cdots,n\}$ is
\[
c_k = c \;=\; \mathbf{1}^T\mathbf{c},\qquad \mathbf{1}=[1\;\;1]^T.
\]

\begin{proof}
 For $m\in\{1,2\}$ and $k\ge 1$, we first define the expected value of accuracy of a set $A_k$ conditioned on the state $Q(t)=m$ as
\begin{equation}
c_k^{(m)}
\;\triangleq\;
\lim_{t \to \infty}\mathbb{E}\!\left[C_{A_k}(t)\,\mathds{1}_{\{Q(t)=m\}}\right].
\label{eq:ck_mode_def}
\end{equation}
Then, we define the overall expected accuracy as 
\begin{align}
    c_k = c_k^{(1)} + c_k^{(2)},
\end{align}
for all $k$. Similar to the analysis for mode-tagged freshness-based accuracy, we use the two following test functions for the average accuracy:
\begin{equation}
\!\!\hat{\psi}_k^{(1)}(Q,\mathbf{Z})
= C_{A_k}\,\mathds{1}_{\{Q=1\}},\!\!
\quad
\hat{\psi}_k^{(2)}(Q,\mathbf{Z})
= C_{A_k}\,\mathds{1}_{\{Q=2\}},\!\!
\label{eq:phi_ck_modes}
\end{equation}
so that we have $c_k^{(m)} = \lim_{t \to \infty} \mathbb{E}[\hat{\psi}_k^{(m)}]$. Then, in a similar fashion to the balance equations for $F_{A_k}^{(1)}$ and $F_{A_k}^{(2)}$, we can write the following balance equations for $\hat{\psi}_k^{(1)}(Q,\mathbf{Z})= C_{A_k}\,\mathds{1}_{\{Q=1\}}$ as:
\begin{align}\label{eq:SHS_C_A_k}
\lim_{t\to\infty}\!\mathbb{E}\Big[\!
&\Big(\!\!-\!\!q_{12}C_{A_k}\!\!
+\!\alpha_k\Big(\tfrac{1\!-\!C_j}{k}\Big)
\!\!+\!\!(\beta_k\!+\!\gamma_k)\!\Big(\tfrac{F_{\{i,j\}}\!-\!C_j}{k}\Big)\!\!\Big)\notag\\[-0.2em]
&\cdot \mathds{1}_{\{Q=1\}}
+\Big(\!q_{21}(1\!-\!C_{A_k})\!\Big)\mathds{1}_{\{Q=2\}}
\!\Big]=0.
\end{align}
After taking the expectations and using the fact that $\frac{d}{dt}\,\mathbb{E}\left[\hat{\psi}\bigl(Q,\mathbf{Z}\bigr)\right] = 0$, we have 
\begin{align}
\!\!\!q_{21}(\pi_2\!-\!c_k^{(2)})\!-\!q_{12}c_k^{(1)}
\!+\!\tfrac{\lambda_s}{n}(\pi_1\!-\!c_1^{(1)})
\!+\!\lambda(f_2^{(1)}\!\!-\!c_1^{(1)})\!=\!0.\!\!
\label{eq:C_balance_mode1_raw}
\end{align}
By following similar steps for $\hat{\psi}_k^{(2)}\!(Q,\mathbf{Z})\!\!=\!\! C_{A_k}\mathds{1}_{\{\!Q\!=\!2\}}$, we get
\begin{align}
\!\!\!q_{12}(\pi_1\!-\!c_k^{(1)})\!-\!q_{21}c_k^{(2)}
\!+\!\tfrac{\lambda_s}{n}(\pi_2\!-\!c_1^{(2)})
\!+\!\lambda(f_2^{(2)}\!\!-\!c_1^{(2)})\!=\!0.\!\!
\label{eq:C_balance_mode2_raw}
\end{align}
We note that $C_{A_k}(t)= \frac{1}{k}\sum_{\ell=1}^{k}C_\ell(t)$ as provided in (\ref{eq:c_ak_defn}). Since the network is fully connected and symmetric, every node in the network experiences the same average accuracy. Because of this reason, all $c_k^{(m)}$ for $k\geq 1$ are equal to each other, that is, $c_k^{(m)} = c^{(m)},$ for all $k \in \{1,\dots,n\},$ and $m\in\{1,2\}$. Then, the accuracy equations in (\ref{eq:C_balance_mode1_raw}) and (\ref{eq:C_balance_mode2_raw}) simplify to
\begin{align}
\!\!\!q_{21}(\pi_2\! -\! c^{(2)})
\!-\! q_{12} c^{(1)}
\!+\! \frac{\lambda_s}{n}\bigl(\pi_1 \!-\! c^{(1)}\bigr)
\!+\! \lambda\bigl(f_2^{(1)} \!\!- \!c^{(1)}\bigr)\! \!=& 0,\!\!\!\label{eq:C_mode_coupled_system1}
\\
\!\!\!q_{12}(\pi_1 \!- \!c^{(1)})
\!- \!q_{21} c^{(2)}
\!+\! \frac{\lambda_s}{n}\bigl(\pi_2 \!- \!c^{(2)}\bigr)
\!+ \!\lambda\bigl(f_2^{(2)} \!\!- \!c^{(2)}\bigr) \!\!=& 0.\!\!\!\label{eq:C_mode_coupled_system2}
\end{align}

Define the mode-tagged average accuracy vector \(\mathbf{c}\triangleq
\begin{bmatrix}
c^{(1)},
c^{(2)}
\end{bmatrix}^T\). Then \eqref{eq:C_mode_coupled_system1}--\eqref{eq:C_mode_coupled_system2} can be written compactly as the form provided in~\eqref{eq:c_linear_system}.

Consequently, the average accuracy of an individual node (equivalently, that of the set $A_k$ for any $k \in \{1,\cdots,n\}$) under a binary CTMC is given by $c \triangleq c^{(1)}+c^{(2)} = \mathbf{1}^T\mathbf{c}$ and by symmetry $c_k=c$ for all $k\ge 1$. This concludes the proof.
\end{proof}
\end{theorem}

Next, we investigate several limits that follow directly from Theorem~\ref{thm:avg_c_theorem}, starting with the case where the source can push new information to the network without any rate constraint. 
\begin{corollary}\label{cor:ls_infty_fixed_lambda}
Fix $\lambda\ge 0$ and let $\mu=\lambda_s/n$. Let $\mathbf{c}=[c^{(1)},c^{(2)}]^T$ be the unique solution of
\eqref{eq:c_linear_system}. Then, as $\lambda_s$ (or equivalently, $\mu$) grows large, we have
\begin{equation}\label{eq:cor_ls_infty}
\lim_{\lambda_s\to\infty}c^{(m)}=\pi_m,\; m\in\{1,2\},
\end{equation}
and as a result, $\lim_{\lambda_s\to\infty}c=\pi_1+\pi_2=1.$ Here, $\boldsymbol{\pi}\triangleq[\pi_1,\pi_2]^T$ is the stationary distribution of the source CTMC.
\end{corollary}
\begin{proof}
From Theorem~\ref{thm:avg_c_theorem}, we have $\bar W\,\mathbf{c}=\mathbf{b}+\lambda\mathbf{f}_2$ with
$\bar W=\mu I_2+\tilde W$ where
$\tilde W\triangleq\begin{bmatrix}q_{12}+\lambda & q_{21}\\ q_{12} & q_{21}+\lambda\end{bmatrix}$.
If we divide both sides by $\mu$, we obtain
\[
\Bigl(I_2+\frac{1}{\mu}\tilde W\Bigr)\mathbf{c}
=
\boldsymbol{\pi}+\frac{\lambda}{\mu}\mathbf{f}_2,
\]
since $\mathbf{b}/\mu$ approaches  $\boldsymbol{\pi}$ as $\mu$ grows arbitrarily large, we have $\lim_{\mu \to \infty} \mathbf{c} = \boldsymbol{\pi}$ and hence $\lim_{\mu \to \infty}  c= 1$.
\end{proof}
Corollary~\ref{cor:ls_infty_fixed_lambda} is a direct result of the assumption that the source always holds fresh and accurate information. As the source push rate rises arbitrarily, this fresh and accurate information dominates all gossip events, forcing all the nodes in the network to be accurate.

In the following corollary, we investigate how the average accuracy of the network behaves if the source stops sending updates to the network.
\begin{corollary}\label{cor:ls_zero_fixed_lambda}
Fix $\lambda>0$ and let $\mu=\lambda_s/n$. Then, as $\lambda_s$ (or equivalently, $\mu$) approaches 0, the mode-tagged accuracy vector $\mathbf{c}$ converges to $\mathbf{c}_0\triangleq\lim_{\lambda_s\to 0}\mathbf{c}$, which is the unique solution of
\begin{equation*}
\begin{bmatrix}
q_{12} \!+ \!\lambda & q_{21}\\
q_{12} & q_{21}\! + \!\lambda
\end{bmatrix}
\mathbf{c}_0
\!=\!
\begin{bmatrix}
q_{21}\pi_2\\
q_{12}\pi_1
\end{bmatrix}
\!+\!
\lambda\,\mathbf{f}_{2,0},
~\text{where }
\mathbf{f}_{2,0}\triangleq \lim_{\lambda_s\to 0}\mathbf{f}_2.
\end{equation*}
Consequently, $\lim_{\lambda_s\to 0}c=\mathbf{1}^T\mathbf{c}_0$.
\end{corollary}
\begin{proof}
Let $\mu$ get arbitrarily close to $0$ in the linear system in \eqref{eq:c_linear_system}. The limit $\mathbf{f}_{2,0}$
is obtained by taking $\alpha_k=\tfrac{k\lambda_s}{n} = 0$ in the recursion that defines $\mathbf{f}_2$, given in \eqref{eq:fk_matrix_recursion}.
\end{proof}
Intuitively, Corollary~\ref{cor:ls_zero_fixed_lambda} can be interpreted as gossip events dominating the information flow in the network when there are no new packets entering the network. Under such conditions, each node's accuracy is determined entirely by the accuracy of the freshest node within the subset of analysis, since that node's content is guaranteed to be accepted by all other nodes if there are no new packets arriving from the source. In this case, even though the source does not transmit fresh or accurate information to the nodes, the nodes may still hold accurate information due to state flips at the source.

We now characterize how the average accuracy behaves when nodes can exchange information without rate limits.
\begin{corollary}\label{cor:lambda_infty_fixed_ls}
Fix $\lambda_s>0$ (equivalently, $\mu>0$). Let $\mathbf{c}$ solve \eqref{eq:c_linear_system}.
Then, as $\lambda$ grows arbitrarily large, we have
\begin{equation*}
\lim_{\lambda\to\infty}\mathbf{c}=\lim_{\lambda\to\infty}\mathbf{f}_2
=
\begin{bmatrix}
f_n^{(1)}\\
f_n^{(2)}
\end{bmatrix},
\qquad
\lim_{\lambda\to\infty}c=\lim_{\lambda\to\infty}\mathbf{1}^T\mathbf{c}=f_n,
\end{equation*}
where $f_n^{(1)}$ and $f_n^{(2)}$ are given in~\eqref{eq:boundary_solution}, and $f_n$ can be obtained from the sum of these two quantities.
\end{corollary}
\begin{proof}
Let us divide \eqref{eq:c_linear_system} by $\lambda$ to obtain:
\[
\left(I_2+\frac{1}{\lambda}
\begin{bmatrix}
q_{12}+\mu & q_{21}\\
q_{12} & q_{21}+\mu
\end{bmatrix}\right)\mathbf{c}
=
\mathbf{f}_2+\frac{1}{\lambda}\mathbf{b}.
\]
Letting $\lambda$ grow large yields $\lim_{\lambda \to \infty}\mathbf{c} = \mathbf{f}_2$. Moreover, in the recursion for $\mathbf{f}_k$ given in~\eqref{eq:fk_matrix_recursion}, letting $\lambda$, and hence $\beta_k$, grow arbitrarily large for each $k\in\{1,\dots,n-1\}$, forces $\mathbf{f}_k-\mathbf{f}_{k+1}$ to approach $0$; and iterating gives $\lim_{\lambda \to \infty}\mathbf{f}_2 = \mathbf{f}_n$.
Finally,
$\lim_{\lambda \to \infty}c=\mathbf{1}^T\mathbf{f}_n=f_n$.
\end{proof}
As $ \lambda$ grows, each node can communicate with another almost instantly. As a result,  Corollary~\ref{cor:lambda_infty_fixed_ls} suggests that the entire network acts as a single node. Since the gossiping protocol is based on freshness, the average accuracy is then entirely dominated by the accuracy of the freshest node within the network, whose information content spreads throughout the network instantly. Finally, we investigate the case where nodes stop gossiping altogether in the following corollary.
\begin{corollary}\label{cor:lambda_zero_fixed_ls}
Fix $\lambda_s>0$ (equivalently, $\mu>0$). In the no-gossip limit, as $\lambda$ approaches $0$, the mode-tagged accuracy vector satisfies
\begin{equation}\label{eq:cor_lambda_zero_system}
\begin{bmatrix}
q_{12} + \mu & q_{21}\\
q_{12} & q_{21} + \mu
\end{bmatrix}
\lim_{\lambda\rightarrow 0}\mathbf{c}
=
\begin{bmatrix}
q_{21}\pi_2+\mu\pi_1\\
q_{12}\pi_1+\mu\pi_2
\end{bmatrix},
\end{equation}
and has the following closed-form solution
\begin{equation}\label{eq:cor_c_lambda0_closed}
\begin{aligned}
\lim_{\lambda\rightarrow 0}\mathbf{c}
&=
\frac{1}{s(s+\mu)}
\begin{bmatrix}
q_{21}(q_{21}+\mu)\\
q_{12}(q_{12}+\mu)
\end{bmatrix},\\[3pt]
\lim_{\lambda\rightarrow 0}c
&=
\lim_{\lambda\rightarrow 0}\mathbf{1}^T\mathbf{c}
=
\frac{\mu (q_{12}+q_{21}) + q_{12}^2+q_{21}^2}{(q_{12}+q_{21})(q_{12}+q_{21}+\mu)}.
\end{aligned}
\end{equation}
\end{corollary}
\begin{proof}
As $\lambda$ get arbitrarily close to $0$ in \eqref{eq:c_linear_system},  $\bar W$ loses the $\lambda I_2$ term and the right-hand side of \eqref{eq:c_linear_system} reduces to $\mathbf{b}$. If we solve the resulting $2\times 2$ linear system explicitly, using $\pi_1=\frac{q_{21}}{q_{12}+q_{21}} $ $= 1-\pi_2$, we obtain the closed form expressions in (\ref{eq:cor_c_lambda0_closed}).
\end{proof}
Intuitively, if there are no gossip events in the network, each node's accuracy will be governed entirely by how fast the source is transitioning between states, and how fast the source can send information to each node. Aside from these two source dynamics, there is no way in which accurate information, or any sort of information for that matter, can propagate throughout the network.

\begin{remark}
For a single node, that is, a subset $A_k$ with size $k=1$, the accuracy of the freshest node is equal to the average accuracy, since there is only a single node in the subset. This means that $f_1 = c_1$ under steady--state conditions. Furthermore, since $c_1=c_k, \; \forall k$ as mentioned earlier, solving the recursion for $f_1$ provides a full characterization of the average accuracy within the network. We provide the analysis in this section for completeness and also to give the explicit asymptotic limits in Corollaries~\ref{cor:ls_infty_fixed_lambda}--\ref{cor:lambda_zero_fixed_ls}.
\end{remark}

Next, we consider a special case in which the binary CTMC is symmetric, allowing us to reduce the matrix-valued recursion derived in this section to a scalar-valued one.
\subsection{The Special Case of Symmetric Binary Source CTMCs}
In this subsection, we consider the case where the source's CTMC is symmetric with a generator matrix given by 
\begin{align*}
\mathbf{Q}=\left[\begin{array}{cc} -q & q \\[2pt] q & -q \end{array} \right].
\end{align*}
With this transition matrix, it is easy to see that the steady state distribution of the symmetric CTMC is equal to $\pi_1 = \pi_2 = 0.5$.  

For this special case, the source can be thought of generating status updates following a Poisson process with equivalent rate $\lambda_E'$ obtained from the transition rates as:
\begin{equation}
    \lambda_E' 
    = \pi_1\, q + \pi_2\, q
    = q.
\label{fig:lambda_eff}
\end{equation}
As a result, the SHS model can be simplified such that it fits the form $(\mathbf{C}(t), \mathbf{X}(t)) \in \mathbb{R}^{2n}$, where
$\mathbf{C}(t) = [C_1(t), \ldots, C_n(t)]$ and 
$\mathbf{X}(t) = [X_1(t), \ldots, X_n(t)]$. In other words, when the source's CTMC is symmetric, we do not need to explicitly follow the source's state $Q(t)$ as in the previous subsections. This newly constructed system obeys the differential equation \((\dot{\mathbf C}(t),\, \dot{\mathbf X}(t)) = \mathbf{0}_{2n}\) and has a single discrete mode. The set of transitions reduces to
\begin{equation}
\mathcal{E}
= \{(0,0)\}
  \;\cup\; \{(0,j) : j \in \mathcal{N}\}
  \;\cup\; \{(i,j) : i,j \in \mathcal{N}\},
\label{eq:set_of_transitions_symmetric}
\end{equation}
which consists of $i)$ the event $(0,0)$ representing the state inversions at the source, $ii)$ the event $(0,j)$ representing the source sending updates to node $j$, and $iii)$ the event $(i,j)$ representing node $i$ sending an update to node $j$, i.e.,  gossiping between nodes $i$ and $j$. The rates of these events are given by
\begin{equation}
\lambda_{ij} =
\begin{cases}
\lambda_E',            & i = 0,\; j = 0, \\
\dfrac{\lambda_s}{n},  & i = 0,\; j \in \mathcal{N}, \\
\dfrac{\lambda}{n-1},  & i, j \in \mathcal{N}.
\end{cases}
\label{eq:lambda_map_symmetric}
\end{equation}
We can now analyze the accuracy of the freshest node in set $A_k$ for a symmetric binary source, denoted by $F_{A_k, \text{sym}}$, and the accuracy of set $A_k$ denoted by $C_{A_k, \text{sym}}$, without their mode-tagged counterparts due to the source transition rate symmetry. Using the reset maps given in~\eqref{eq:C_Ak_update} and~\eqref{eq:F_Ak_update}, alongside the SHS generator identity, we obtain 
\begin{equation}
\sum_{(i,j)\in\mathcal{E}} \lambda_{ij}
\Bigl(
    \mathbb{E}[\psi(\phi_{i,j})]
    -
    \mathbb{E}[\psi]
\Bigr) = 0.
\label{eq:symmetric_balance}
\end{equation}
Then, we can obtain the steady-state values $f_{k, \text{sym}} = \lim_{t \to \infty} \mathbb{E}[F_{A_k, \text{sym}}(t)]$ as follows:
\begin{equation}
\lambda_E' (1 - 2 f_{k, \text{sym}})
+ \alpha_k (1 - f_{k, \text{sym}})
+ \beta_k (f_{k+1, \text{sym}} - f_{k, \text{sym}})= 0,
\label{eq:fk_balance}
\end{equation}
where $\alpha_k$ and $\beta_k$ are defined in (\ref{eq:alpha_beta_def}). Finally, rearranging the terms in \eqref{eq:fk_balance} yields the backward recursion given by
\begin{equation}
f_{k,\text{sym}} =
\frac{
\lambda_E'
+ \alpha_k
+ \beta_k f_{k+1,\text{sym}}
}{
2\lambda_E'
+ \alpha_k
+ \beta_k
}\cdot
\label{eq:fk_recursive}
\end{equation}

Next, defining $C_{\emptyset,\text{sym}} = 0$, since an empty set cannot contain accurate information, a similar set of derivations yields the steady-state characterization of $C_{A_k,\text{sym}}$, given by $c_{k,\text{sym}} = \lim_{t \to \infty} \mathbb{E}[C_{A_k,\text{sym}}(t)]$ Additionally, we denote the steady-state value of $F_{\{i,j\},\text{sym}}$ as $f_{2,\text{sym}} = \lim_{t \to \infty} \mathbb{E}[F_{\{i,j\},\text{sym}}(t)]$. Thus, we have:  
\begin{equation*}
\begin{aligned}
\!\lambda_E'(1 - 2c_k)+
\!\lambda k\left(\dfrac{f_{2,\text{sym}} - c_{1,\text{sym}}}{k} \right) \!+ \!\alpha_k
\left(\!\dfrac{1\! -\! c_{1,\text{sym}}}{k}\! \right) \!=\! 0.
\end{aligned}
\end{equation*}
Collecting all terms in $c_{k,\text{sym}}$ and solving for $c_{k,\text{sym}}$ alongside $f_{1,\text{sym}}=c_{1,\text{sym}} = c_{k,\text{sym}}$ gives
\begin{equation}
c_{k,\text{sym}} \;=\;
\frac{
\lambda_E'
+ \frac{\lambda_s}{n}
+ \lambda f_{2,\text{sym}}
}{
2\lambda_E' + \frac{\lambda_s}{n} + \lambda
},
\label{eq:ck_recursive2}
\end{equation}
for all $k\!\geq \!1$. As it can be seen from~\eqref{eq:fk_recursive} and~\eqref{eq:ck_recursive2}, the recursive formulations for both freshness-based accuracy and average accuracy reduce to scalar-valued ones instead of matrix-valued ones, since there is no need to explicitly track the source state via indicator functions. Next, we extend our accuracy analysis for a source having a general $M>2$-state CTMC. 

\section{The General Multi-State Information Source}\label{sect:n_state}
We now consider the case when the information source has a CTMC with $M>2$ states, i.e., $Q(t)\in\{1,\cdots,M\}$. Different from the binary case, we can no longer assume that both the average accuracy $C_{A_k}$ and the freshness-based accuracy $F_{A_k}$ simply flip after a source inversion, since the source no longer \emph{inverts}, but rather changes to another state. Therefore, after a source transition, the accuracy of the freshest node in set $A_k$, i.e., $F_{A_k}$, depends on whether or not the freshest node holds the contents of the next state the source transitions into, i.e., $Q'(t)$. Thus, under the new source definition, $F_{A_k}(t)$ can be re-expressed as follows:
\begin{equation}
F_{A_k}(t) \triangleq \mathds{1}_{\{ R_k(t) = Q(t) \}},
\label{eq:F_Ak_def}
\end{equation}
where $R_k(t) = S_{\arg\min_{\, j \in A_k} X_j(t)}$ represents the content of the freshest node within the $k$-node subset $A_k$ of the network. Next, we define a new joint process as
\begin{equation}
Z_k(t) \triangleq \bigl(Q(t),\, R_k(t)\bigr),
\label{eq:z_k_defn}
\end{equation}
which takes values in the finite state space
\begin{equation}
\mathcal{S}
\triangleq
\{(q,r) : q\in\{1,\ldots,M\},\; r\in\{1,\ldots,M\}\}.
\label{eq:state_space}
\end{equation}

Let $\pi^{(k)}_{q,r}$ denote the stationary distribution of $Z_k(t)$, which is given by
\begin{equation}
\pi^{(k)}_{q,r}
\triangleq
\mathbb{P}(Q=q,\; R_k=r),
\label{eq:stationary_dist_qr}
\end{equation}
for all $(q,r)\in\mathcal{S}$. The resulting steady-state freshness-based accuracy, $f_k = \lim_{t \to \infty} \mathbb{E}[F_{A_k}(t)]$, can be obtained by
\begin{equation}
f_k
\triangleq
\mathbb{P}(R_k = Q)
=
\sum_{q=1}^M \pi^{(k)}_{q,q}.
\label{eq:f_k_summation_defn}
\end{equation}

For each $k\in\{1,\ldots,n\}$, the process $Z_k(t)$ is a CTMC on $\mathcal{S}$. We list the states of $Z_k(t)$ as $\{(1,1),(1,2),\cdots,(1,M),(2,1),\cdots,(2,M),\cdots,(M,M)\}$ and define a generator matrix $Q_k$ which has dimensions $M^2 \times M^2$.
The transitions in $Q_k$ arise from three cases: $i)$ source-to-source transitions of $Q(t)$, $ii)$ source-to-$A_k$ update events, and $iii)$  gossip events originating from nodes outside $A_k$. As a result, the matrix $Q_k$ will consist of three distinct components, $Q_k^{\text{src}},Q_k^{\text{push}}$, and $Q_k^{\text{out}}$ all of which have dimensions $M^2 \times M^2$, respectively. We now describe the contribution of each case. 

A source transition from state $i$ to state $j$ with $i\neq j$ does not change $R_k(t)$. Therefore, the corresponding diagonal and off-diagonal entries of the source component of the generator, $Q_k^{\mathrm{src}}$, are given by 
\begin{align}
Q_k^{\mathrm{src}}\bigl((i,r),(j,r)\bigr)
=
\begin{cases}
q_{i j},           &\text{for } j\neq i,\\
-\sum_{\substack{j=1 \\ j\neq i}}^{M}
q_{i j},  &\text{for }  j=i,
\end{cases}
\label{eq:source_to_source_trans}
\end{align}
for all $r$ and the remaining entries of $Q_k^{\mathrm{src}}$ are equal to 0. We also note that each row of $Q_k^{\mathrm{src}}$ sums to zero. 

Next, the total update rate of pushes from the source into the set $A_k$ is $\alpha_k$ as described in~\eqref{eq:alpha_beta_def}.
A push event sets the freshest content equal to the source state. If the freshest content is already equal to the source state, no change occurs. Hence, the push component of the generator, $Q_k^{\mathrm{push}}$, has off-diagonal entries
\begin{equation}
Q_k^{\mathrm{push}}\bigl((i,r),(i,i)\bigr)
=
\alpha_k\, \mathds{1}_{\{r\neq i\}},
\label{eq:q_k_push_defn}
\end{equation}
and diagonal entries
\begin{equation}
Q_k^{\mathrm{push}}\bigl((i,r),(i,r)\bigr)
=
-\alpha_k\, \mathds{1}_{\{r\neq i\}},
\label{eq:q_k_push_diag}
\end{equation}
and the remaining entries of $Q_k^{\mathrm{push}}$ are given as 0. Gossip events arriving from outside $A_k$ into the subset occur at rate
$\beta_k$. As in the binary case, we need to enlarge the set of nodes to obtain a characterization of $F_{A_k}$ under gossiping events. When we enlarge the set of interest from $k$ to $k+1$ nodes, the freshest node's content changes from $r$ to $r'$ according to the conditional distribution of $R_{k+1}$ given the source state. We therefore define the conditional probabilities 
\begin{equation}
p^{(k+1)}_{r\mid q}
\triangleq
\mathbb{P}(R_{k+1}=r \mid Q=q),
\label{eq:k_plus_1_conditional_defn}
\end{equation}
for $q,r\in\{1,\ldots,M\}$. Using the stationary distribution of the $(k+1)$-node process, we have
\begin{equation}
p^{(k+1)}_{r\mid q}
=
\frac{\pi^{(k+1)}_{q,r}}
     {\pi^{(k+1)}_{q}},
\label{eq:k_plus_1_conditional_pi}
\end{equation}
where $\pi^{(k+1)}_{q}
\triangleq
\sum_{r'=1}^M \pi^{(k+1)}_{q,r'}.$
Then, the outsider gossip into $A_k$ induces transitions in the gossip component of the generator, $Q_k^{\mathrm{out}}$, of the form
\begin{equation}
Q_k^{\mathrm{out}}\bigl((q,r),(q,r')\bigr)
=
\beta_k p^{(k+1)}_{r'\mid q},
\qquad
r'\neq r,
\label{eq:q_k_out_defn}
\end{equation}
with the corresponding diagonal entries
\begin{equation}
Q_k^{\mathrm{out}}\bigl((q,r),(q,r)\bigr)
\!=\!
-\beta_k
\sum_{\substack{r'=1 \\ r'\neq r}}^{M}
p^{(k+1)}_{r'\mid q}
\!=\!
-\beta_k
\bigl(1 \!-\! p^{(k+1)}_{r\mid q}\bigr).\!\!
\label{eq:q_k_out_diag}
\end{equation}
Combining these three cases, the total generator $Q_k$ of $Z_k(t)$ is given by
\begin{equation}
Q_k
=
Q_k^{\mathrm{src}}
+
Q_k^{\mathrm{push}}
+
Q_k^{\mathrm{out}}.
\label{eq:Qk_total}
\end{equation}
For each fixed $k$, the stationary distribution $\pi^{(k)}$ satisfies 
\begin{equation}
\boldsymbol{\pi}^{(k)} Q_k = \boldsymbol{0}, 
\label{eq:pi_k_stationary_eq}
\end{equation}
 and $\sum_{(q,r)\in\mathcal{S}} \pi^{(k)}_{q,r} = 1$ where $\boldsymbol{\pi}^{(k)}$ is the row vector consisting of
$\boldsymbol{\pi}^{(k)}=
\bigl[
\pi^{(k)}_{1,1},\allowbreak \cdots,\allowbreak \pi^{(k)}_{1,M},\allowbreak
\pi^{(k)}_{2,1},\allowbreak \cdots,\allowbreak \pi^{(k)}_{2,M},\allowbreak
\cdots,\allowbreak
\pi^{(k)}_{M,1},\allowbreak \cdots,\allowbreak \pi^{(k)}_{M,M}
\bigr].$ Since the values $p^{(k+1)}_{r\mid q}$ in \eqref{eq:k_plus_1_conditional_defn}–\eqref{eq:k_plus_1_conditional_pi} depend on the stationary distribution for the $(k+1)$-node system, the generator matrices $\{Q_k\}$ and the corresponding stationary distributions can be built recursively by following Algorithm~\ref{alg:backward_Qk_pi_fk}.

\begin{algorithm}[t]
\caption{Backward Construction of $\{Q_k,\boldsymbol{\pi}^{(k)}\}_{k=1}^n$ and Computation of $\{f_k\}_{k=1}^n$}
\label{alg:backward_Qk_pi_fk}
\begin{algorithmic}[1]
\Require Network size $n$; source CTMC state-count $M$, source push rate $\alpha_k$, gossip rate $\beta_k$; stationary equation \eqref{eq:pi_k_stationary_eq}; conditional definitions \eqref{eq:k_plus_1_conditional_defn}--\eqref{eq:k_plus_1_conditional_pi}.
\Ensure Generator matrices $\{Q_k\}_{k=1}^n$, stationary distributions $\{\boldsymbol{\pi}^{(k)}\}_{k=1}^n$, and accuracies $\{f_k\}_{k=1}^n$.

\State \textbf{Initialize at $k=n$:} set $\beta_n \gets 0$.
\State Form $Q_n \gets Q_n^{\mathrm{src}} + Q_n^{\mathrm{push}}$.
\State Solve \eqref{eq:pi_k_stationary_eq} for $k=n$ to obtain $\boldsymbol{\pi}^{(n)}$.

\For{$k = n-1, n-2, \ldots, 1$}
    \State Using $\boldsymbol{\pi}^{(k+1)}$, compute the conditional probabilities
    $p^{(k+1)}_{r\mid q}$ via \eqref{eq:k_plus_1_conditional_defn}--\eqref{eq:k_plus_1_conditional_pi}.
    \State Form $\!Q_k\!$ using $p^{(k+1)}_{r\mid q}$ (and the model construction of $\!Q_k\!$).
    \State Solve \eqref{eq:pi_k_stationary_eq} for the stationary distribution $\boldsymbol{\pi}^{(k)}$.
\EndFor

\For{$k = 1,2,\ldots,n$}
    \State Compute $f_k$ from $\boldsymbol{\pi}^{(k)}$ using \eqref{eq:f_k_summation_defn}.
\EndFor

\State \Return $\{Q_k,\boldsymbol{\pi}^{(k)},f_k\}_{k=1}^n$.
\end{algorithmic}
\end{algorithm}

After the stationary distribution $\pi^{(k)}_{q,r}$ has been obtained for each $k$, the freshness-based accuracy $f_k$ is given by \eqref{eq:f_k_summation_defn}. Note that, as in the binary case, the average accuracy of any subset of size $k$ coincides with the accuracy of a single node, that is,
\begin{equation}
f_1 = c_1 = c_k, \qquad \forall k\in\{1,\ldots,n\},
\label{eq:fk_ck_relation}
\end{equation}
so this joint CTMC model characterizes the average accuracy for any arbitrary subset of the network with cardinality $k$.

We next present a lemma showing that the network’s information-content distribution depends only on the stationary distribution of the source CTMC.
\begin{lemma}\label{lem:number_of_nodes}
The expected number of nodes in an arbitrary subset of the network with cardinality $k$ holding information $q$ is equal to $k\pi_q$ for any given $\lambda,\lambda_s>0$  under steady--state conditions.
\end{lemma}
\begin{proof}
Consider an extension of the SHS model described in~\eqref{eq:hybrid_state},~\eqref{eq:set_of_transitions}, and~\eqref{eq:lambda_map} with $M$ discrete modes corresponding to the $M$ source states.
 
Let $\bar{\psi}_k(\mathbf{C},\mathbf{X}) = N_{A_k}^{(q)}$ be a time-invariant test function that gives the number of nodes holding information $q$ in set $A_k$, formally defined as:
\begin{align}\nonumber\\[-21pt]
    \bar{\psi}_k(\mathbf{C},\mathbf{X}) \triangleq N_{A_k}^{(q)}(t)
= \sum_{j\in A_k}\mathds{1}_{\{S_j(t)=q\}},
\label{eq:n_ak_definition}\\[-21pt]\nonumber
\end{align}
Then, the transition map for $N_{A_k}^{(q)}$ can be written as follows:

\begin{align}\nonumber\\[-21pt]
\!\!N_{A_k}^{(q)\prime}\!\! =\!\!
\begin{cases}
N_{A_k}^{(q)} \!+\! \mathds{1}_{\{Q=q\}} \!- \!H_{\{j\}}^{(q)},
& \!\!\!e\!=\!(0,j), j\!\in\! A_k, \\
N_{A_{k}}^{(q)} \!+\! H_{\{i,j\}}^{(q)} \!- \!H_{\{j\}}^{(q)},
& \!\!\!e\!=\!(i,j), i\!\in\! \mathcal{N},  j \!\in\! A_k,\!\!\! \\
N_{A_k}^{(q)},
&\!\!\! \text{otherwise,}
\end{cases}
\label{eq:NAkq_update}\\[-21pt]\nonumber
\end{align}
where $H_{A_k}^{(q)}(t) = \mathds{1}_{\{S_{\operatorname*{arg\,min}_{j \in A_k} X_j(t)} = q\}}$. Note that $H_{A_k}^{(q)}(t)$ reduces to $H_{j}^{(q)}(t)= \mathds{1}_{\{S_j(t)=q\}}$ for a subset consisting of a single node. This allows us to model changes to the indicator function term in~\eqref{eq:n_ak_definition} using recursive expressions for $H_{A_k}$. In order to characterize $n_k^{(q)}= \lim_{t \to \infty}\mathbb{E}[N_{A_k}^{(q)}(t)]$, we first need to characterize $h_k^{(q)} = \lim_{t \to \infty}\mathbb{E}[H_{A_k}^{(q)}(t)]$, which can be done by using the following reset map:
\begin{align}\nonumber\\[-21pt]
H_{A_k}^{(q)\prime} =
\begin{cases}
\mathds{1}_{\{Q=q\}},
& e=(0,j),\; j\in A_k, \\
H_{A_{k+1}}^{(q)},
& e=(i,j),\; i\in \mathcal{N}\setminus A_k,\; j\in A_k, \\
H_{A_k}^{(q)},
& \text{otherwise.}
\end{cases}
\label{eq:SAkq_update}\\[-21pt]\nonumber
\end{align}
Recognizing the fact that $\lim_{t\to \infty}\mathbb{E}[\mathds{1}_{\{Q=q\}}] = \pi_q$, the steady-state balance equations for $H_{A_k}^{(q)}$ can be written as follows:
\begin{align}\nonumber\\[-21pt]
  h_k^{(q)}
= \frac{
  \dfrac{k\lambda_s}{n}\,\pi_q
  + \dfrac{k(n-k)\lambda}{\,n-1\,}\,h_{k+1}^{(q)}
}{
  \dfrac{k\lambda_s}{n}
  + \dfrac{k(n-k)\lambda}{\,n-1\,}
}.
\label{eq:hkq_recursion}  \\[-21pt]\nonumber
\end{align}
Using the fact that the recursion reduces to $h_n^{(q)} = \pi_q$ for $k=n$, it can be shown that $h_k^{(q)}=\pi_q$ for all $k\in \{1,\cdots,n\}$. Using the result obtained from~\eqref{eq:hkq_recursion}, a recursive expression for $n_k^{(q)}$ can be written as follows:
\begin{align}
0
&=
\frac{k\lambda_s}{n}\bigl(\pi_q-h_1^{(q)}\bigr)
+
k\lambda\bigl(h_2^{(q)}-h_1^{(q)}\bigr).
\label{eq:nk_proof_intermediate_eqn}
\end{align}
From $N_{A_k}^{(q)}$'s definition in~\eqref{eq:n_ak_definition}, we have
\begin{equation}
    n_k^{(q)}=\lim_{t \to \infty}\mathbb{E}[N_{A_k}^{(q)}]=\sum_{j\in A_k}\lim_{t \to \infty}\mathbb{E}[\mathds{1}_{\{S_j=q\}}]=k\,h_1^{(q)}.
\end{equation}
Plugging this result back alongside $\!h_1^{(q)}\!\!\!=\!\!\frac{n_k^{(q)}}{k}\!$ into~\eqref{eq:nk_proof_intermediate_eqn}, we get
\begin{align}
0
&=
\frac{k\lambda_s}{n}\!\left(\pi_q-\frac{n_k^{(q)}}{k}\right)
+
k\lambda\bigl(h_2^{(q)}-h_1^{(q)}\bigr),\\
n_k^{(q)}
&=
k\pi_q+\frac{nk\lambda}{\lambda_s}\bigl(h_2^{(q)}-h_1^{(q)}\bigr).
\label{eq:nk_proof_step}
\end{align}
Finally, using \(h_2^{(q)}=h_1^{(q)}=\pi_q\) in \eqref{eq:nk_proof_step}, we have
\begin{equation}
    {n_k^{(q)}=k\pi_q}.
\end{equation}
This concludes the proof.
\end{proof}

\section{Information Accuracy Due to Freshness}\label{sect:usefulness}
As noted earlier, a node can possess accurate information in two distinct ways. First, the node may hold the most recent update generated at the source, in which case its version age is zero. Second, even if the node does not have the freshest update, that is, its version age is strictly positive, it may still be accurate because subsequent information flips at the source cause the source’s current state to coincide with the information stored at the node. In this case, despite being outdated in terms of version age, the node’s information remains accurate. In this section, we distinguish between the two cases and quantify the resulting information inaccuracy attributable to $i)$ freshness effects and $ii)$ source-state transitions.

\begin{theorem}\label{thm:acc_fresh_thm}
    The fraction of nodes within any given subset of the network with cardinality $k \in \{1,2,\cdots,n-1\}$ that hold strictly fresh and accurate information is given by $G_k = \sum_{m=1}^{M} c_k^{(m)}g_k^{(m)},$ for $ M \in \{2,3,\ldots\}$, where $g_k^{(m)}=
    \frac{\alpha_k\pi_m + \beta_k g_{k+1}^{(m)}}
    {\nu_m+\alpha_k+\beta_k}$.
\begin{proof}
    To begin the proof, we investigate the average accuracy due to freshness and source-state flips using the SHS approach for a binary source CTMC, in a manner similar to that in Section~\ref{sect:binary}. We then extend this SHS model with $M$ different test functions for the general multi-state information source. For the binary case, consider two new mode-tagged test functions to explicitly model how effective information freshness (or similarly, source-state flips) are for accurate information dissemination. We define $\tilde\psi_k^{(m)}(Q, \mathbf Z)
\;=\;
c_k^{(m)}\,\mathds{1}_{\{Q=m, X_{A_k} = 0\}}$ to denote the average accuracy of set $A_k$ due to fresh information, whose steady-state values will be denoted by:
\begin{align}\nonumber\\[-21pt]
\lim_{t \to \infty}\mathbb{E}\big[\tilde\psi_k^{(m)}\big]
\;=\;
c_k^{(m)}\,g_k^{(m)},\\[-21pt]\nonumber
\end{align}
where $g_k^{(m)}
\;\triangleq\;
\lim_{t \to \infty}\mathbb{E}\left[ \mathds{1}_{\{Q=m, X_{A_k}=0\}}\right]
\;=\;
\mathbb{P}(Q=m,\ X_{A_k}=0).$
Since $c_k^{(m)}$ is a constant (already determined by the information accuracy analysis), only $g_k^{(m)}$ needs to be characterized.
Since the analysis for both test functions is similar, we will only provide the analysis of the first mode-tagged function here, in a manner similar to that of the mode-tagged freshness functions' analyses.

Consider the transition $e=(0,0,1 \to 2)$. Before the transition, $\tilde\psi_k^{(1)}(Q,\mathbf Z)=\mathds{1}_{\{Q=1, X_{A_k=0}\}}$ evaluates to $\tilde\psi_k^{(1)}(Q,\mathbf Z)
\;=\;\mathds{1}_{\{X_{A_k}=0\}},$  as the source content indicator gives 1. After the aforementioned transition, the version age of all nodes in the network increases by 1, as given by the transition map in~\eqref{eq:X_Ak_update}. As a result, the indicator function for version age will cause the expression to evaluate to zero, leading to the overall change given by $-\mathds{1}_{\{Q=1,X_{A_k}=0\}}$ term.

Now, consider the source push transition given by $(0,j)$ for $ j \in A_k$. After such a transition, the recipient node will contain accurate and fresh information, leading the test function to equal $\tilde\psi_k^{(1)}(Q,\mathbf Z)
\;=\;
\,\mathds{1}_{\{Q=1\}}\,$ which leads to the overall change of $\mathds{1}_{\{Q=1\}} - \mathds{1}_{\{Q=1, X_{A_k = 0}\}}$ under this transition.

Finally, consider gossiping events given by the transition $ e= (i,j)$ for $ j\in A_k$ and $ i\in \mathcal{N} \setminus A_k$. Under this transition, the source content is not changed, and the version age analysis has to be extended to the super-set $A_{k+1}$ as given by~\eqref{eq:X_Ak_update}. As a result, the overall change of the test function under this transition will be $\mathds{1}_{\{Q=1\}}
\big(
\mathds{1}_{\{X_{A_{k+1}}=0\}}
-
\mathds{1}_{\{X_{A_k}=0\}}
\big)$.

Taking the expectations of all these terms, and adding the rates $\alpha_k$ and $\beta_k$ for source pushes and outsider gossip events respectively, the following SHS balance equation is obtained:
\begin{align}\nonumber\\[-21pt]
g_k^{(1)}
= \frac{\alpha_k \pi_1 + \beta_k g_{k+1}^{(1)}}
       {q_{12} + \alpha_k + \beta_k}.
\label{eq:g_backward_recursion_mode_1}\\[-21pt]\nonumber
\end{align}
With similar derivations, an expression for the second-mode counterpart can be obtained as follows:
\begin{align}\nonumber\\[-21pt]
g_k^{(2)}
= \frac{\alpha_k \pi_2 + \beta_k g_{k+1}^{(2)}}
       {q_{21} + \alpha_k + \beta_k}.
\label{eq:g_backward_recursion_mode_2}\\[-20pt]\nonumber
\end{align}

We define $\mathbf{g}_k \triangleq [g_k^{(1)}, g_k^{(2)}]^T$ and $\boldsymbol{\pi}\triangleq [\pi_1,\pi_2]^T. $ Then, these two equations can be written in terms of a matrix-recursion as the earlier test functions in the form of:
\begin{align}\nonumber\\[-21pt]
\mathbf{g}_k
=
\mathbf{A}_k\,\boldsymbol{\pi}
+
\mathbf{B}_k\,\mathbf{g}_{k+1},
\label{eq:g_matrix_backward_recursion}\\[-21pt]\nonumber
\end{align}
where
\begin{align*}\\[-21pt]
\resizebox{\columnwidth}{!}{$
\begin{aligned}
\mathbf{A}_k
&\triangleq
\begin{bmatrix}
\dfrac{\alpha_k}{q_{12}+\alpha_k+\beta_k} & 0\\
0 & \dfrac{\alpha_k}{q_{21}+\alpha_k+\beta_k}
\end{bmatrix},\quad
\mathbf{B}_k
\triangleq
\begin{bmatrix}
\dfrac{\beta_k}{q_{12}+\alpha_k+\beta_k} & 0\\
0 & \dfrac{\beta_k}{q_{21}+\alpha_k+\beta_k}
\end{bmatrix}.
\end{aligned}
$}\\[-21pt]
\end{align*}

This concludes the analysis for the binary source case. This analysis can be extended to the the general information source with $M$ states, with $M$ distinct mode-tagged test functions of the form $\tilde\psi_k^{(m)}(Q, \mathbf Z)
\;=\;
c_k^{(m)}\,\mathds{1}_{\{Q=m, X_{A_k} = 0\}}$ for $m \in \{1,\dots,M\}$.

Recognizing the pattern from the recursive expressions in~\eqref{eq:g_backward_recursion_mode_1} and~\eqref{eq:g_backward_recursion_mode_2}, the derivation can be generalized using these $M$ distinct test functions as follows:
\begin{align}\nonumber\\[-21pt]
g_k^{(m)} =
\frac{\alpha_k\pi_m + \beta_k g_{k+1}^{(m)}}
{\nu_m+\alpha_k+\beta_k},
\qquad
k=1,\ldots,n-1.
\label{eq:M_state_gk_final}\\[-21pt]\nonumber
\end{align}
where $\nu_m \triangleq \sum_{r\neq m} q_{mr} = 1-q_{mm}.$ As a result, the fraction of nodes in an arbitrary $k$-node subset holding accurate and fresh information can then be obtained by $G_k = \sum_{m=1}^{M} c_k^{(m)}g_k^{(m)},$ for $ M \in \{2,3,\ldots\}$ which concludes the proof.
\end{proof}
\end{theorem}

We now state several limits that arise as a direct consequence of Theorem \ref{thm:acc_fresh_thm}, starting with the fraction of \emph{accurate and stale} nodes within the network.
\begin{corollary}
The fraction of nodes within an arbitrary $k$-node subset that have accurate but stale information is given by $c_k-G_k$, where the values of $c_k$ can be obtained by following the analysis in Sections~\ref{sect:binary} and~\ref{sect:n_state}.
\begin{proof}
Recall that $c_k$ denotes the overall probability (fraction) that set $A_k$ holds accurate information, i.e.,
$c_k = \mathbb{P}(R_k = Q)$, while $G_k$ denotes the probability that it is both accurate and \emph{strictly fresh},
i.e., $G_k = \mathbb{P}(R_k = Q,\; X_{A_k}=0)$.
Since the events $\{X_{A_k}=0\}$ (fresh) and $\{X_{A_k}>0\}$ (stale) are disjoint and exhaustive, we have
\begin{align*}\\[-21pt]
\{R_k=Q\}
=
\{R_k=Q,\;X_{A_k}=0\}
\;\cup\;
\{R_k=Q,\;X_{A_k}>0\}.\\[-21pt]
\end{align*}
Taking probabilities yields
$
c_k = G_k + \mathbb{P}(R_k = Q,\; X_{A_k}>0),
$
hence the fraction of accurate-but-stale nodes is
$
\mathbb{P}(R_k = Q,\; X_{A_k}>0) = c_k - G_k.
$
\end{proof}
\end{corollary}

We now investigate how information accuracy due to freshness behaves if the source can push updates to the network without rate limits.
\begin{corollary}\label{cor:high_ls}
For any $k\in\{1,\ldots,n\}$ and mode $m\in\{1,\ldots,M\}$, if $\lambda_s$, and hence $\alpha_k$, grows arbitrarily large,
then
\begin{align}\nonumber\\[-21pt]
\lim_{\lambda_s\to\infty} g_k^{(m)} = \pi_m,~
\lim_{\lambda_s\to\infty} G_k\! =\! \sum_{m=1}^M c_k^{(m)}\pi_m\!= \!\sum_{m=1}^M \pi_m^2.\\[-21pt]\nonumber
\end{align}
\begin{proof}
From \eqref{eq:M_state_gk_final}, we have
\begin{align*}\\[-21pt]
g_k^{(m)}=\frac{\alpha_k\pi_m+\beta_k g_{k+1}^{(m)}}{\nu_m+\alpha_k+\beta_k}.\\[-21pt]
\end{align*}
Divide numerator and denominator by $\alpha_k$ and take the limit as $\alpha_k$ grows arbitrarily large:
\begin{align*}\nonumber\\[-21pt]
\lim_{\alpha_k \to \infty}g_k^{(m)}
=
\lim_{\alpha_k \to \infty}\frac{\pi_m + (\beta_k/\alpha_k) g_{k+1}^{(m)}}{(\nu_m/\alpha_k)+1+(\beta_k/\alpha_k)}
\;=
\pi_m,\\[-21pt]\nonumber
\end{align*}
since $0\le g_{k+1}^{(m)}\le 1$. The claim for $G_k=\sum_m c_k^{(m)}g_k^{(m)}$ follows from the high $\lambda_s$ limit of $c_k$ given in~\eqref{eq:cor_ls_infty}.
\end{proof}
\end{corollary}
Intuitively, when the source pushes arbitrarily fast, nodes are \emph{almost surely fresh}, so the event $\{X_{A_k}=0\}$ loses selectivity and the mode probability is dominated by the source’s stationary behavior $\pi_m$. 

We now investigate when the source stops updating the network.
\begin{corollary}\label{cor:low_ls}
For any $k\in\{1,\ldots,n\}$ and $m\in\{1,\ldots,M\}$, we have
\begin{align}\nonumber\\[-21pt]
\lim_{\lambda_s\to 0^+} g_k^{(m)} = 0,
\qquad
\lim_{\lambda_s\to 0^+} G_k = 0.\\[-21pt]\nonumber
\end{align}
\begin{proof}
At $k=n$ there are no outside nodes, hence $\beta_n=0$ and the recursion reduces to
\begin{align*}\\[-21pt]
\lim_{\alpha_n \to 0 }  g_n^{(m)}= \lim_{\alpha_n \to 0 }\frac{\alpha_n\pi_m}{\nu_m+\alpha_n}=0.\\[-21pt]
\end{align*}
Now apply \eqref{eq:M_state_gk_final} backward for $k=n-1,n-2,\ldots,1$. Substituting $\alpha_k =  0$, the recursion becomes
\begin{align*}\\[-21pt]
g_k^{(m)}=\frac{\beta_k}{\nu_m+\beta_k}\,g_{k+1}^{(m)},\\[-21pt]
\end{align*}
and the multiplicative factor $\beta_k/(\nu_m+\beta_k)$ is finite. Since $g_n^{(m)}$ approaches $0$, it follows that all preceding $g_k^{(m)}$ approach $0$ as well, hence $G_k$ approaches $0$.
\end{proof}
\end{corollary}
Intuitively, as the network becomes disconnected from the source, all nodes’ version ages drift to infinity, so the zero-age indicator $\mathds{1}_{\{X_{A_k}=0\}}$ forces $g_k^{(m)}=\mathbb{P}(Q=m,X_{A_k}=0)$ to vanish.

Corollary~\ref{cor:high_lambda} characterizes the behavior of $G_k$ when all nodes gossip instantaneously, i.e., in the limit as $\lambda \to \infty$.
\begin{corollary}\label{cor:high_lambda}
For any $\lambda_s>0$ (and hence, $\alpha_k>0$) and any $m \in \{1,\cdots,M\}$,
\begin{align}\nonumber\\[-21pt]
\lim_{\lambda\to\infty}\bigl(g_k^{(m)}-g_{k+1}^{(m)}\bigr)=0,
\qquad k=1,\ldots,n-1,\\[-21pt]\nonumber
\end{align}
and thus $\!g_1^{(m)}\!=\!\cdots\!=\!g_n^{(m)}\!$ in the limit as $\lambda$ grows large.
\end{corollary}
\begin{proof}
From \eqref{eq:M_state_gk_final}, we have $g_k^{(m)}=\frac{\alpha_k\pi_m+\beta_k g_{k+1}^{(m)}}{\nu_m+\alpha_k+\beta_k}.$ Next, we define
$D_k^{(m)}\!\triangleq\!\nu_m+\alpha_k+\beta_k,$ $
w_k^{(m)}\!\triangleq\!\frac{\alpha_k}{D_k^{(m)}},$ $
v_k^{(m)}\!\triangleq\!\frac{\beta_k}{D_k^{(m)}},$ and $
r_k^{(m)}\!\triangleq\!\frac{\nu_m}{D_k^{(m)}}.$
Then, we can write $g_k^{(m)}$ as
\begin{align*}\\[-21pt]
g_k^{(m)}=w_k^{(m)}\pi_m+v_k^{(m)}g_{k+1}^{(m)}+r_k^{(m)}\cdot 0,\\[-21pt]
\end{align*}
By using $w_k^{(m)}+v_k^{(m)}+r_k^{(m)}=1,$ we obtain
\begin{align*}\\[-21pt]
g_k^{(m)}-g_{k+1}^{(m)}
=
w_k^{(m)}\!\left(\pi_m-g_{k+1}^{(m)}\right)-r_k^{(m)}g_{k+1}^{(m)}.\\[-21pt]
\end{align*}
Using $0\le g_{k+1}^{(m)}\le 1$ and $0\le \pi_m\le 1$, we get 
\begin{align*}\\[-21pt]
\left|g_k^{(m)}-g_{k+1}^{(m)}\right|
\le w_k^{(m)}+r_k^{(m)}
= \frac{\alpha_k+\nu_m}{\nu_m+\alpha_k+\beta_k}.\\[-21pt]
\end{align*}
As $\lambda$, and hence $\beta_k$, grows arbitrarily large, the RHS tends to $0$ and thus we have $\lim_{\lambda\to\infty}\bigl(g_k^{(m)}-g_{k+1}^{(m)}\bigr)=0.$
\end{proof}

Intuitively, when gossiping is arbitrarily fast, the distinction between a subset $A_k$ and its super-sets disappears, so the network behaves (from the perspective of freshness) like a single node, which is consistent with the findings of the previous sections.

Finally, we investigate the case when gossiping within the network ceases, and only source pushes remain.
\begin{corollary}\label{cor:gn_closed_form}
Fix $\lambda_s>0$ and consider $\lambda=0$ (so $\beta_k=0$ for all $k$). Then, for every
$k\in\{1,\ldots,n\}$ and $m\in\{1,\ldots,M\}$, we have
\begin{align} \nonumber\\[-21pt]
\lim_{\lambda \to 0}g_k^{(m)}
=
\mathbb{P}\!\left(Q=m,\;X_{A_k}=0\right)
=
\frac{\alpha_k\,\pi_m}{\nu_m+\alpha_k}.
\label{eq:gk_no_gossip}\\[-21pt]\nonumber
\end{align}
Consequently,
\begin{align}\nonumber\\[-21pt]
\lim_{\lambda \to 0}G_k
\!=\!
\sum_{m=1}^M \lim_{\lambda \to 0}c_k^{(m)}\,g_k^{(m)}
\!=\!
\sum_{m=1}^M \frac{\alpha_k\,\pi_m}{\nu_m+\alpha_k}\lim_{\lambda \to 0}c_k^{(m)}.\!\!
\label{eq:Gk_no_gossip}\\[-21pt]\nonumber
\end{align}
\end{corollary}
\begin{proof}
With $\lambda=0$, no gossip transitions occur. As $\lambda$ gets arbitrarily close to 0,~\eqref{eq:M_state_gk_final} gives
\begin{align*}\\[-21pt]
\lim_{\lambda_k \to 0} g_k^{(m)}
=
\frac{\alpha_k\pi_m}{\nu_m+\alpha_k},
\qquad k=1,\ldots,n,\\[-21pt]
\end{align*}
which proves \eqref{eq:gk_no_gossip}. Plugging this expression into $G_k=\sum_{m=1}^M c_k^{(m)}g_k^{(m)}$
yields \eqref{eq:Gk_no_gossip}.  Since a closed form solution for $c_k^{(m)}$ does not exist for the general $M$-state information source, the limit $\lim_{\lambda \to 0}c_k^{(m)}$ can be evaluated by setting an arbitrarily low gossip rate $\lambda$ in Algorithm~\ref{alg:backward_Qk_pi_fk}, and solving the recursion. If the source CTMC is binary, a closed form solution for $\lim_{\lambda \to 0}c_k^{(m)}$ exists and is provided in Corollary~\ref{cor:lambda_zero_fixed_ls}.
\end{proof}
When $\lambda=0$, nodes in $A_k$ can become fresh only via direct source pushes, since there is no peer-to-peer spreading to ``import'' freshness from outside the set. Thus, $g_k^{(m)}=\mathbb{P}(Q=m, X_{A_k}=0)$ equals the steady-state fraction of time the source is in mode $m$ (namely, $\pi_m$) times the probability that a push arrives before the source leaves mode $m$, which is $\alpha_k/(\nu_m+\alpha_k)$.

In the next section, we present our numerical results, and compare the theoretical results derived in Sections~\ref{sect:binary}--~\ref{sect:usefulness} against simulations.

\section{Numerical Results and Asymptotic Analysis Insights}\label{sect:numerical}

In this section, we provide numerical results for a) binary information source, b) multi-state information source, and c) comparison of information accuracy due to freshness and source inversions. 

\subsection{Results for Binary Information Source}
In our first experimental result, we simulate the gossip network described in Section~\ref{sect:model}, for a total of $2,250,000$ time-steps, with $n=10$ nodes and $M=2$ source-states, and sweep over the values of $\lambda$ and $\lambda_s$ respectively, while keeping the other variables constant. The source generator matrix is chosen as 
\begin{align*}\nonumber\\[-21pt]
  \mathbf{Q} = \begin{bmatrix}
-0.25 & 0.25 \\
0.75 & -0.75
\end{bmatrix}.  \nonumber\\[-21pt]
\end{align*} 
In Fig.~\ref{fig:asymmetric_binary_graphs}(a), we provide how freshness-based accuracy of a node ($f_1)$ and mode-tagged freshness-based
accuracy, i.e., $f_1^{(m)}$, change with respect to $\lambda_s$. In Fig.~\ref{fig:asymmetric_binary_graphs}(b), we provide how freshness-based accuracy, i.e., $f_k$, changes with respect to $\lambda$ for $k \in \{1,3,5,10\}$. We note that $f_1$ also gives the average-accuracy, $c_k = c_1$, for any arbitrary subset with the cardinality $k$.  The parameter that is not swept over is fixed to 1 for each plot; that is, $\lambda = 1$ for Fig.~\ref{fig:asymmetric_binary_graphs}(a) and $\lambda_s=1$ for Fig.~\ref{fig:asymmetric_binary_graphs}(b).
\begin{figure}[t]
    \centering
    \subfloat[]{%
        \includegraphics[width=0.485\columnwidth]{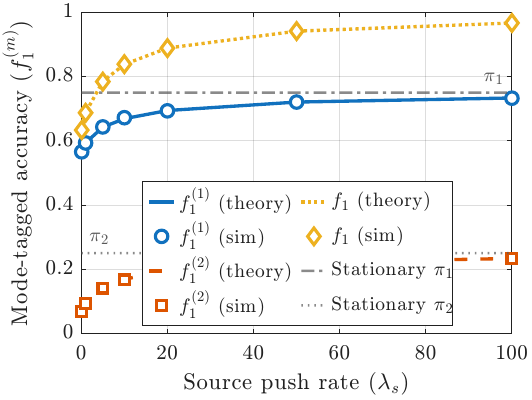}%
        \label{fig:asymmetric_lambda_s}%
    }\hfill
    \subfloat[]{%
        \includegraphics[width=0.485\columnwidth]{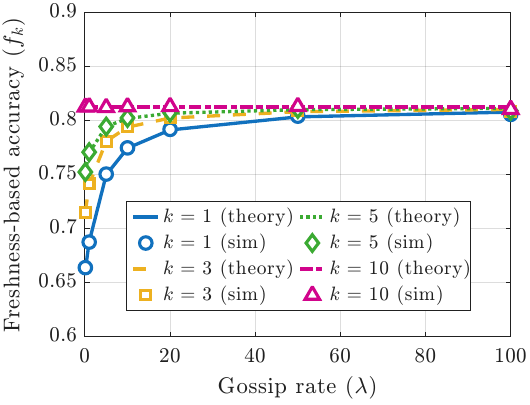}%
        \label{fig:asymmetric_lambda}%
    }\vspace{-0.2cm}
    \caption{Freshness-based accuracy vs.\ (a) $\lambda_s$ for $k \in \{1,3,5,10\}$ and (b) $\lambda$ for $k=1$.}
    \label{fig:asymmetric_binary_graphs}
    \vspace{-0.5cm}
\end{figure}
To understand Fig.~\ref{fig:asymmetric_binary_graphs}(a), we focus our attention on the recursive $\mathbf{f}_k $ derivation provided in (\ref{eq:fk_matrix_recursion}) with $W_k$ and $\mathbf{v}_k$ expressions in (\ref{eq:A_k_B_k_def}). As $\lambda_s$ grows large, $\alpha_k$ also increases proportionally which makes the matrix $W_k$ approach $\alpha_kI_2$ and vector $\mathbf{v}_k$ gets closer to $\alpha_k \boldsymbol{\pi}$. Since the final component of~(\ref{eq:fk_matrix_recursion}), that is, $V_k \mathbf{f}_{k+1}$ with $ V_k =\beta_k I_2$,  only consists of finite elements, the vector $\mathbf{f}_k $ converges to $ \boldsymbol{\pi}$ as $\lambda_s$ grows. Intuitively, if the source always pushes its current information into the network, the nodes will always carry the fresh and accurate information, and therefore $f_{k}^{(1)}$
and $f_{k}^{(2)}$ will take on the stationary distributions of the source for their respective mode tags. Furthermore, since $f_k = f_k^{(1)} +f_k^{(2)}$, in the limit, $\lim_{\alpha_k \to \infty}f_k = \pi_1 + \pi_2 = 1$, meaning the freshest node of any subset will always be accurate. This trend is observed in Fig.~\ref{fig:asymmetric_binary_graphs}(a) and is proved in Corollary~\ref{cor:ls_infty_fixed_lambda}. As $\lambda$ grows large, and thus $\beta_k$ grows proportionally, we have $\dfrac{1}{\beta_k}W_k$ converging to $I_2$ and $\dfrac{1}{\beta_k}\mathbf{v}_k$ getting closer to $0$. As a result, we have $
\lim_{\beta_k \to \infty} \mathbf{f}_k = \mathbf{f}_{k+1}.$ Intuitively, if all nodes communicate with one another at every time instant, the entire network behaves as a singular node. Such a trend can be seen in Fig.~\ref{fig:asymmetric_binary_graphs}(b) and and is proved in Corollary~\ref{cor:lambda_infty_fixed_ls}.

From Figs.~\ref{fig:asymmetric_binary_graphs}(a) and~\ref{fig:asymmetric_binary_graphs}(b), it can be seen that increasing $\lambda_s$ is much more effective for raising the average accuracy than increasing $\lambda$. The reason for this is that as $\lambda$ increases while keeping $\lambda_s$ constant, the nodes are able to obtain the freshest information in the network, but the freshest information in the network will depend on the source's rate $\lambda_s$. On the other hand, if we increase $\lambda_s$ while keeping $\lambda$ constant, every node in the network is still able to obtain the most up-to-date information from the source; and thus, the nodes’ accuracy converges to $\!1$.  

Our second simulation results shown in Figs.~\ref{fig:f_and_c_symmetric_plots}(a) and ~\ref{fig:f_and_c_symmetric_plots}(b) illustrate the behavior of the average accuracy under a symmetric binary source. Similar to the earlier simulation result, the parameter that is not swept is fixed to 1, that is, $\lambda = 1$ for Fig.~\ref{fig:f_and_c_symmetric_plots}(a) and $\lambda_s = 1$ for Fig.~\ref{fig:f_and_c_symmetric_plots}(b), respectively. These figures verify that the steady-state assumption of all nodes behaving identically, and the average accuracy being equal to that of the one-node subset for all subset cardinalities $k\in\{1,2,\cdots,n\}$. Here, we also see that $\lambda_s$ is much more effective at increasing average accuracy compared to $\lambda$.

\begin{figure}[t]
    \centering
    \subfloat[]{%
        \includegraphics[width=0.485\columnwidth]{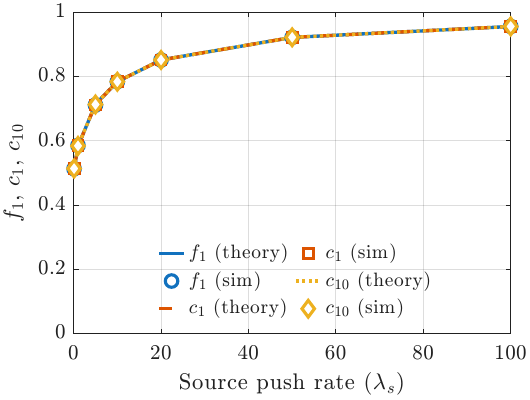}%
        \label{fig:symmetric_vs_lambda_s}%
    }\hfill
    \subfloat[]{%
        \includegraphics[width=0.485\columnwidth]{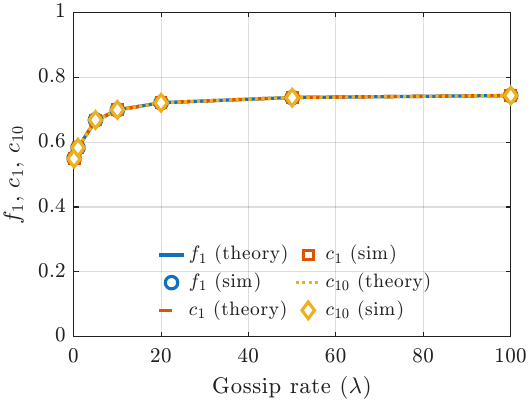}%
        \label{fig:symmetric_vs_lambda}%
    }
    \caption{The freshness-based and average accuracy of a node ($f_{1,\text{sym}}$ and $c_{1,\text{sym}}$, respectively) for a symmetric CTMC versus $\lambda_s$ and $\lambda$.
    \vspace{-5mm}}
    \label{fig:f_and_c_symmetric_plots}
\end{figure}
\subsection{Results for General Multi-State Information Source}
In our third simulation results with the general information source model, we consider a source that has a $M=4$-state CTMC with its generator matrix given by
\begin{align}\nonumber\\[-21pt]
\mathbf{Q}=
\begin{bmatrix}
-1.00 & 0.25 & 0.15 & 0.60 \\
0.40 & -0.85 & 0.30 & 0.15 \\
0.20 & 0.85 & -1.15 & 0.10 \\
0.15 & 0.25 & 0.45 & -0.85
\end{bmatrix}.
\label{eq:Q_4x4_matrix}\\[-21pt]\nonumber
\end{align}
As with the binary CTMC case, the network size is taken to be $n=10$, and a total of $2,250,000$ time-steps were used to simulate steady-state conditions.

\begin{figure}[t]
    \centering
    \subfloat[]{%
        \includegraphics[width=0.485\columnwidth]{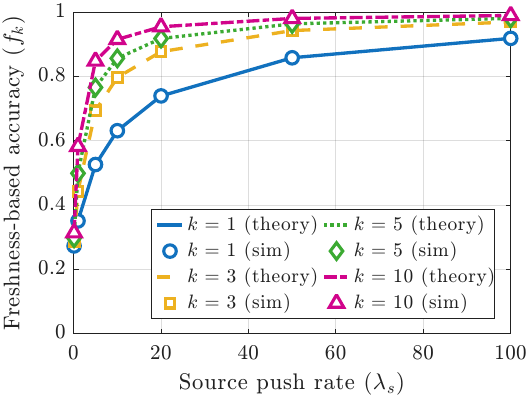}%
        \label{fig:n_state_lambda_s}%
    }\hfill
    \subfloat[]{%
        \includegraphics[width=0.485\columnwidth]{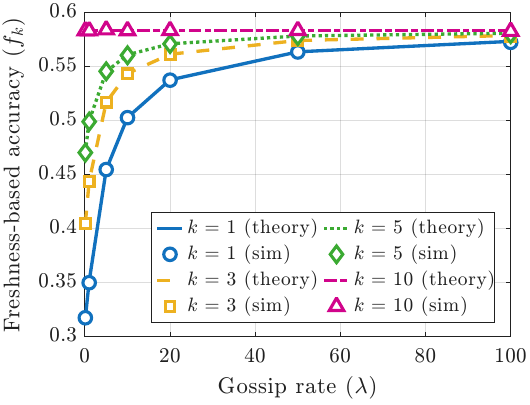}%
        \label{fig:n_state_lambda}%
    }
     \vspace{-0.25cm}
    \caption{Freshness-based accuracy for $k \in \{1,3,5,10\}$ (a) versus $\lambda_s$, and (b) versus $\lambda$.}
    \label{fig:n_State_fk}
    \vspace{-0.5cm}
\end{figure}

Figs.~\ref{fig:n_State_fk}(a) and ~\ref{fig:n_State_fk}(b) depict the simulation results for the general multi-state information source with $4$ states. As $\lambda$ increases, the average accuracy for each $k$ value converges to the same value, highlighting how the network behaves as a single node for high gossip rates. A representative example for this case would be how the average accuracy for $k=n$ does not change with $\lambda$, and all other $f_k$ values for $k \neq n$ approach $f_{n}$. As we increase $\lambda_s$, every node in the network is able to get the freshest information, and as a result, their accuracy get closer to 1 as shown in Fig.~\ref{fig:n_State_fk}(a). Thus, the observation made for the $M=2$ state source is also validated for a general $M$-state source.

\begin{figure}[t]
  \centering
  \includegraphics[width=0.55\linewidth]{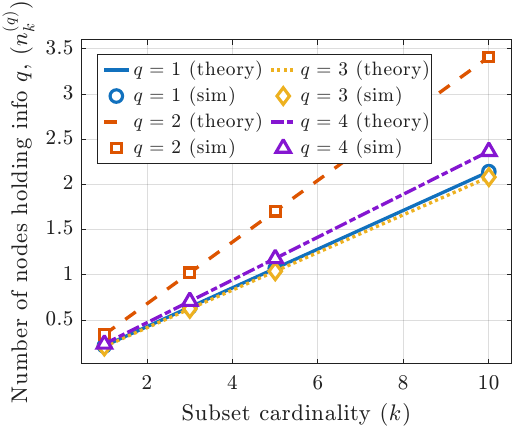}
  \vspace{-0.35cm}
  \caption{Number of nodes holding information $q$ vs. set cardinality $k$ under a 4-state CTMC source}
  \label{fig:nk_vs_k}
  \vspace{-0.45cm}
\end{figure}

\begin{figure}[t]
    \centering
    \subfloat[]{%
        \includegraphics[width=0.485\columnwidth]{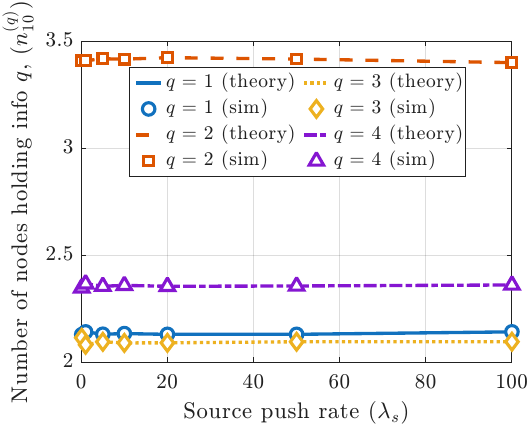}%
        \label{fig:nk_vs_lambda_s}%
    }\hfill
    \subfloat[]{%
        \includegraphics[width=0.485\columnwidth]{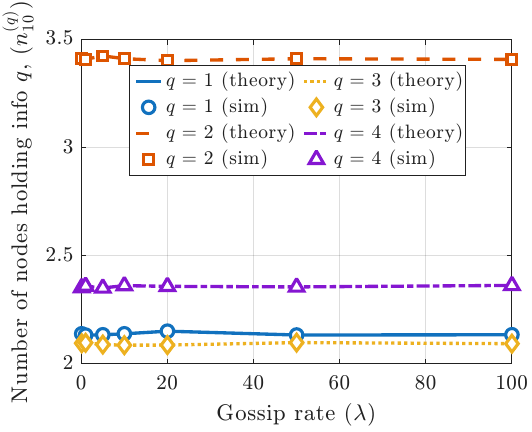}%
        \label{fig:nk_vs_lambda}%
    } \vspace{-0.25cm}
    \caption{Number of nodes holding information $q$, $n_k^{(q)}$, vs.\ (a) source push rate $\lambda_s$ and (b) gossip rate $\lambda$.}
    \label{fig:nk_figure}
    \vspace{-0.35cm}
\end{figure}

We also provide simulation results for the expressions obtained in Lemma~\ref{lem:number_of_nodes},  which states that the expected fraction of nodes holding information $q$ any subset of the network with cardinality $k$ is given by $\pi_qk$ under steady--state conditions. For sweeps against subset cardinality $k$ shown in Fig.~\ref{fig:nk_vs_k}, the source push rate and the gossip rate were fixed to 1, $\lambda_s=\lambda=1$. For $n_k^{(q)}$ sweeps against the source push rate and the gossip rate, the parameter not being swept was fixed to 1, i.e., $\lambda=1$, for Fig.~\ref{fig:nk_figure}(a) and $\lambda_s = 1$ for Fig.~\ref{fig:nk_figure}(b). As expected, the number of nodes holding information $q$ increases linearly with subset cardinality $k$, shown in Fig.~\ref{fig:nk_vs_k}. Furthermore, Figs.~\ref{fig:nk_figure}(a) and (b) show how the value for the number of nodes in the entire network holding information $q$, that is $n_{n}^{(q)}$ stays constant as both $\lambda_s$ and $\lambda$ vary, respectively.
\subsection{Results for the Information Accuracy due to  Freshness}\label{sect:asymp_gk}
\begin{figure}[t]
    \centering
    \subfloat[]{%
        \includegraphics[width=0.485\columnwidth]{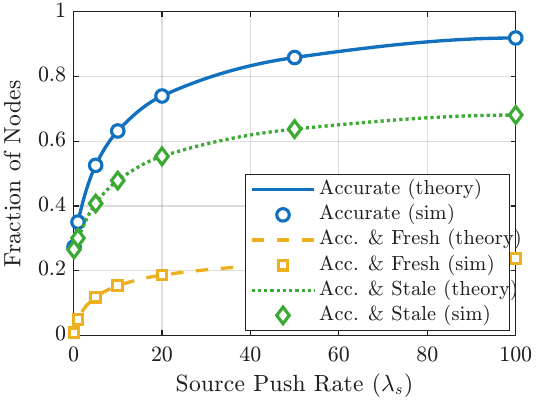}%
        \label{fig:acc_fresh_lambda_s}%
    }\hfill
    \subfloat[]{%
        \includegraphics[width=0.485\columnwidth]{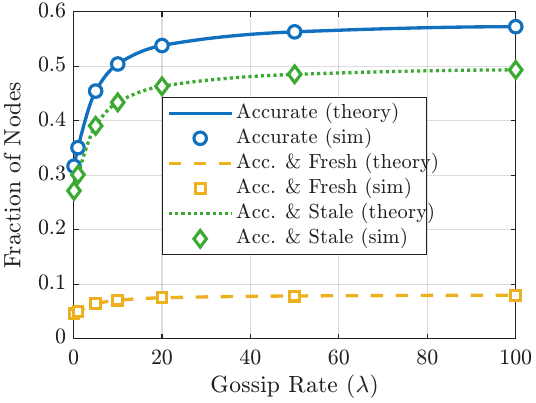}%
        \label{fig:acc_fresh_lambda}%
    }\vspace{-0.25cm}
    \caption{Fraction of nodes (in the entire network) holding fresh and accurate information vs.\ (a) source push rate $\lambda_s$ and (b) gossip rate $\lambda$.}
    \label{fig:acc_fresh}
    \vspace{-0.45cm}
\end{figure}
In our final simulation result, we consider the fraction of nodes that have accurate and \emph{strictly fresh} information ($G_{k=n}$) and observe the impact of near-instantaneous gossip rates on accuracy. As with earlier simulation results, the network size was taken to be $n=10$, and the parameter not being swept, that is, $\lambda$ for Fig.~\ref{fig:acc_fresh}(a) and $\lambda_s$ for Fig.~\ref{fig:acc_fresh}(b), is fixed to 1. Note that $G_k \neq G_j, \; k\neq j, \; k,j \in \mathcal{N}$, so we cannot pick the $k$ value arbitrarily. Since we want to observe the behavior of the entire network, we choose $k=n$ for both plots. A large gap can be seen between Fig.~\ref{fig:acc_fresh}(a) and Fig.~\ref{fig:acc_fresh}(b) for the fraction of nodes that are accurate and fresh. Even when the gossip rate is very fast, only $\approx 0.08$ of nodes are both accurate and strictly fresh, whereas the total accurate (fresh or stale) node fraction is much larger ($\approx 0.57$). Thus, high $\lambda$ mainly synchronizes the nodes to the \emph{same} information, but does not guarantee that this information is \emph{current}; meaning freshness is ultimately bottle-necked by the source push rate $\lambda_s$. This can also be observed by the fact that the \emph{accurate and stale} node fraction noticeably increases as the gossip rate grows larger.

\section{Conclusion and Future Work}\label{sect:discuss_conc}
In this paper, we characterized information accuracy in \emph{timeliness-based} fully connected gossip networks, where nodes accept packets solely based on information freshness (VAoI). Under this rule, nodes may hold the \emph{freshest} packet while still disagreeing with the source state. To quantify this mismatch, we introduced \emph{average accuracy} and \emph{freshness-based accuracy}, and develop an SHS-based steady-state analysis for binary CTMC sources. We derived a matrix-valued backward recursion that characterizes \emph{mode-tagged} freshness-based accuracy with an explicit boundary solution, enabling computation for all subset sizes. We  extended this framework to general $M$-state sources via a joint-CTMC approach, quantifying the contributions of direct source pushes and peer-to-peer gossip. Numerical and asymptotic results showed that increasing the source push rate $\lambda_s$ is significantly more effective than increasing the gossip rate $\lambda$ in improving steady-state accuracy. Future work includes extensions to general network topologies, multiple source processes, optimal rate allocation, and generalized acceptance rules for broader information dissemination settings.

\bibliographystyle{IEEEtran}
\bibliography{citations}

@article{He2020,
   author = {S He and H Shin and S Xu and A Tsourdos},
   journal = {Information Fusion},
   month = {February},
   pages = {21-43},
   publisher = {Elsevier},
   title = {Distributed estimation over a low-cost sensor network: A review of state-of-the-art},
   volume = {54},
   year = {2020}
}

@article{Villalba2009,
   author = {L Villalba and A Orozco and A Cabrera and C Abbas},
   issue = {11},
   journal = {Sensors},
   month = {October},
   pages = {8399-8421},
   publisher = {Molecular Diversity Preservation International (MDPI)},
   title = {Routing protocols in wireless sensor networks},
   volume = {9},
   year = {2009}
}

@article{Akkaya2005,
   author = {K Akkaya and M Younis},
   issue = {3},
   journal = {Ad Hoc Networks},
   month = {May},
   pages = {325-349},
   publisher = {Elsevier},
   title = {A survey on routing protocols for wireless sensor networks},
   volume = {3},
   year = {2005}
}

@INPROCEEDINGS{Abolhassani2021,
  author={Abolhassani, B. and Tadrous, J. and Eryilmaz, A. and Yeh, E.},
  booktitle={IEEE INFOCOM 2021 - IEEE Conference on Computer Communications}, 
  title={Fresh Caching for Dynamic Content},
month={May},
  year={2021},
  volume={},
  number={},
  pages={1-10},
  doi={10.1109/INFOCOM42981.2021.9488731}}

@Article{Abd-Elmagid2023,
AUTHOR = {Abd-Elmagid, Mohamed A. and Dhillon, Harpreet S.},
TITLE = {Distribution of the Age of Gossip in Networks},
JOURNAL = {Entropy},
VOLUME = {25},
month={February},
YEAR = {2023},
NUMBER = {2},
ARTICLE-NUMBER = {364},
PubMedID = {36832729},
ISSN = {1099-4300}
}

@INPROCEEDINGS{Yates2021_SPAWC,
  author={Yates, R. D.},
  booktitle={2021 IEEE 22nd International Workshop on Signal Processing Advances in Wireless Communications (SPAWC)}, 
  title={Timely Gossip},
month = {September},
  year={2021},
  volume={},
  number={},
  pages={331-335},
  doi={10.1109/SPAWC51858.2021.9593163}}

@article{Yates2019,
   author = {R. D. Yates and S. K. Kaul},
   issue = {3},
   journal = {IEEE Transactions on Information Theory},
   month = {March},
   pages = {1807-1827},
   title = {The Age of Information: Real-Time Status Updating by Multiple Sources},
   volume = {65},
   year = {2019}
}

@article{Hespanha2007,
  title={Modelling and analysis of stochastic hybrid systems},
  author={Hespanha, J. P.},
  journal={IEE Proceedings-Control Theory and Applications},
  volume={153},
  number={5},
  pages={520--535},
  month={September},
  year={2006},
  publisher={IET}
}

@inproceedings{v2x_paper,
  title={A spatiotemporal analysis of age of information at {V2X}-enabled intersection},
  author={Wang, Yaying and Tan, Guoping and Zhou, Siyuan and Yang, Dengsong and Ni, Baili},
  booktitle={Proc. IEEE HPCC/DSS/SmartCity/DependSys},
  pages={901--907},
  month= {December},
  year={2021},
  organization={IEEE}
}

@INPROCEEDINGS{Kaswan2023b,
  author={Kaswan, P. and Ulukus, S.},
  booktitle={GLOBECOM 2023 - 2023 IEEE Global Communications Conference}, 
  title={Information Mutation and Spread of Misinformation in Timely Gossip Networks}, 
month = {December},
  year={2023},
  volume={},
  number={},
  pages={3560-3565},
  doi={10.1109/GLOBECOM54140.2023.10437596}}

@article{cached_harvest,
  author  = {E. Delfani and N. Pappas},
  title   = {Version Age-Optimal Cached Status Updates in a Gossiping Network With Energy Harvesting Sensor},
  journal = {IEEE Transactions on Communications},
  month={April},
  year    = {2025},
  volume  = {73},
  number  = {4},
  pages   = {2344--2360},
  doi     = {10.1109/TCOMM.2024.3471969}
}

@INPROCEEDINGS{Srivastava2023_grid,
  author={Srivastava, A. and Ulukus, S.},
  booktitle={2023 59th Annual Allerton Conference on Communication, Control, and Computing (Allerton)}, 
  title={Age of Gossip on a Grid},
month = {September},
  year={2023},
  volume={},
  number={},
  pages={1-8},
  doi={10.1109/Allerton58177.2023.10313495}}

@ARTICLE{Buyukates2022,
  author={Buyukates, B. and Bastopcu, M. and Ulukus, S.},
  journal={IEEE Journal on Selected Areas in Information Theory}, 
  title={Version Age of Information in Clustered Gossip Networks},
 month={March},
  year={2022},
  volume={3},
  number={1},
  pages={85-97},
  doi={10.1109/JSAIT.2022.3159745}}

@INPROCEEDINGS{Mitra2023,
  author={Mitra, P. and Ulukus, S.},
  booktitle={IEEE INFOCOM 2023 - IEEE Conference on Computer Communications Workshops (INFOCOM WKSHPS)}, 
  title={Timely Opportunistic Gossiping in Dense Networks},
month = {May},
  year={2023},
  volume={},
  number={},
  pages={1-6},
  doi={10.1109/INFOCOMWKSHPS57453.2023.10225855}}

@article{Sanghavi2007,
   author = {S. Sanghavi and B. Hajek and L. Massoulie},
   issue = {12},
   journal = {IEEE Transactions on Information Theory},
   month = {December},
   pages = {4640-4654},
   title = {Gossiping With Multiple Messages},
   volume = {53},
   year = {2007}
}

@article{Deb2006,
   author = {S. Deb and M. Medard and C. Choute},
   issue = {6},
   journal = {IEEE Transactions on Information Theory},
   month = {June},
   pages = {2486-2507},
   title = {Algebraic gossip: a network coding approach to optimal multiple rumor mongering},
   volume = {52},
   year = {2006}
}

@article{Shah2008,
   author = {D. Shah},
   issue = {1},
   journal = {Foundations and Trends in Networking},
   pages = {1-125},
   title = {Gossip Algorithms},
   volume = {3},
   year = {2008}
}

@inproceedings{Boyd2005,
   author = {S. Boyd and A. Ghosh and B. Prabhakar and D. Shah},
   booktitle = {IEEE Annual Joint Conference: Infocom, IEEE Computer and Communications Societies},
   month = {March},
   title = {Gossip algorithms: Design, analysis and applications},
   year = {2005}
}

@inproceedings{Kaul2011,
   author = {S K Kaul and M Gruteser and V Rai and J Kenney},
   booktitle = {IEEE Infocom},
   month = {March},
   title = {Minimizing age of information in vehicular networks},
   year = {2011}
}

@article{Yates2020,
   author = {R. D. Yates and Y. Sun and D. R. Brown and S. K. Kaul and E. Modiano and S. Ulukus},
   issue = {5},
   journal = {IEEE Journal on Selected Areas in Communication},
   month = {May},
   pages = {1183-1210},
   title = {Age of Information: An Introduction and Survey},
   volume = {39},
   year = {2020}
}

@book{Sun2022,
   author = {Y. Sun and I. Kadota and R. Talak and E. Modiano},
   publisher = {Springer Nature},
   title = {Age of information: A new metric for information freshness},
   year = {2022}
}

@ARTICLE{Kaswan_survey,
  author={Kaswan, P. and Mitra, P. and Srivastava, A. and Ulukus, S.},
  journal={IEEE Transactions on Communications}, 
  title={Age of Information in Gossip Networks: A Friendly Introduction and Literature Survey}, 
month={August},
  year={2025},
  volume={73},
  number={8},
  pages={6200-6220},
  doi={10.1109/TCOMM.2024.3522043}}

@article{luo2024exploitingdatasignificanceremote,
  title={Exploiting Data Significance in Remote Estimation of Discrete-State {Markov} Sources}, 
      author={Luo, J. and Pappas, N.},
  journal={Available on arXiv:2406.18270},
      year={2024},
    month= {June}
}

@INPROCEEDINGS{Kaswan_timestomp,
  author={Kaswan, P. and Ulukus, S.},
  booktitle={2022 IEEE Information Theory Workshop (ITW)}, 
  title={Susceptibility of Age of Gossip to Timestomping}, 
month={November},
  year={2022},
  volume={},
  number={},
  pages={398-403},
  doi={10.1109/ITW54588.2022.9965757}}

@INPROCEEDINGS{Yates_Age_of_Gossip,
  author={Yates, R. D.},
  booktitle={2021 IEEE International Symposium on Information Theory (ISIT)}, 
  title={The Age of Gossip in Networks}, 
month ={July},
  year={2021},
  volume={},
  number={},
  pages={2984-2989},
  doi={10.1109/ISIT45174.2021.9517796}}

@ARTICLE{yates_aoi_moments,
  author={Yates, Roy D.},
  journal={IEEE Transactions on Information Theory}, 
  title={The Age of Information in Networks: Moments, Distributions, and Sampling}, 
  month={September},
  year={2020},
  volume={66},
  number={9},
  pages={5712-5728},
  doi={10.1109/TIT.2020.2998100}}

@ARTICLE{Maatouk2020,
  author={Maatouk, A. and Kriouile, S. and Assaad, M. and Ephremides, A.},
  journal={IEEE/ACM Transactions on Networking}, 
  title={The Age of Incorrect Information: A New Performance Metric for Status Updates}, 
  year={2020},
month = {October},
  volume={28},
  number={5},
  pages={2215-2228},
  doi={10.1109/TNET.2020.3005549}}

@article{Saurav2025,
  title={Monitoring State Transitions in {Markovian} Systems with Sampling Cost}, 
      author={Saurav, K. and Shroff,  N. B. and  Liang, Y.},
  journal={Available on arXiv:2510.22327},
      year={2025},
    month= {October}
}

@ARTICLE{Salimnejad2025,
  author={Salimnejad, M. and Kountouris, M. and Ephremides, A. and Pappas, N.},
  journal={IEEE Transactions on Communications}, 
  title={Age of Information Versions: A Semantic View of {Markov} Source Monitoring},
month= {December},
  year={2025},
  volume={73},
  number={12},
  pages={14486-14502},
  doi={10.1109/TCOMM.2025.3616209}}

@article{Luo2025,
  title={On the role of age and semantics of information in remote estimation of {Markov} sources},
  author={Luo, J. and Pappas, N.},
  journal={Available on arXiv:2507.18514},
month = {July},
  year={2025}
}

@article{Cosandal2024,
  title={Multi-threshold {AoII}-optimum sampling policies for {CTMC} information sources},
  author={Cosandal, I. and Akar, N. and Ulukus, S.},
  journal={Available on arXiv:2407.08592},
  month={July},
  year={2024}
}

@ARTICLE{nail_hoca_freshness,
  author={Akar, Nail and Ulukus, Sennur},
  journal={IEEE Transactions on Communications}, 
  title={Query-Based Sampling of Heterogeneous CTMCs: Modeling and Optimization With Binary Freshness}, 
  year={2024},
  volume={72},
  number={12},
  pages={7705-7714},
  doi={10.1109/TCOMM.2024.3409542}}

@article{nail_hoca_20_ocak,
      title={Utilizing the Perceived Age to Maximize Freshness in Query-Based Update Systems}, 
      author={Sahan Liyanaarachchi and Sennur Ulukus and Nail Akar},
        journal={Available on arXiv:2601.14075},
        month={January},
      year={2026} 
}

@article{Cosandal2025,
  title={Minimizing Functions of Age of Incorrect Information for Remote Estimation},
  author={Cosandal, I. and Ulukus, S. and Akar, N.},
  journal={Available on arXiv:2504.10451},
 month={April},
  year={2025}
}

@inproceedings{Luo2024,
  author    = {J. Luo and N. Pappas},
  title     = {Minimizing the Age of Missed and False Alarms in Remote Estimation of {Markov} Sources},
  booktitle = {Proc. ACM MobiHoc},
  year      = {2024},
  pages     = {381--386},
  doi       = {10.1145/3641512.3690161}
}

@ARTICLE{Kaswan_mutations,
  author={Kaswan, P. and Ulukus, S.},
  journal={IEEE Transactions on Communications}, 
  title={Misinformation Spread in Gossip Networks: The Influence of Transmission Mutations}, 
month={September},
  year={2025},
  volume={73},
  number={9},
  pages={7711-7723},
  doi={10.1109/TCOMM.2025.3552298}}

@INPROCEEDINGS{bastopcu_binary_fresh,
  author={Bastopcu, M. and Buyukates, B. and Ulukus, S.},
  booktitle={2021 IEEE Globecom Workshops (GC Wkshps)}, 
  title={Gossiping with Binary Freshness Metric},
month={December},
  year={2021},
  volume={},
  number={},
  pages={1-6},
  doi={10.1109/GCWkshps52748.2021.9682174}}

@article{maranzatto_gossip,
  title={Information Degradation and Misinformation in Gossip Networks},
  author={Maranzatto, T. J. and Srivastava, A. and Ulukus, S.},
  journal={Available on arXiv:2501.13086},
month={January},
  year={2025}
}

\end{document}